\documentclass[journal]{IEEEtran}
\usepackage{amsfonts}
\usepackage{amssymb}
\usepackage[cmex10]{amsmath}
\usepackage{graphicx}
\usepackage{subfigure}
\usepackage{verbatim}
\usepackage{cite}
\usepackage{tabularx,booktabs}
\usepackage{amsthm}
\usepackage{booktabs}
\usepackage[ruled,lined]{algorithm2e}
\usepackage{multirow}
\usepackage{setspace}
\usepackage{stfloats}
\usepackage{fancyhdr}
\usepackage{float}
\usepackage{algorithmic}
 \usepackage{enumerate}

\IEEEoverridecommandlockouts

\newcommand{\be}{\begin{equation}}
\newcommand{\ee}{\end{equation}}

\newcommand{\bpi}{\mbox{\boldmath{ $\pi $}}}

\newcommand{\bII}{{\bf I }}

\newcommand{\bX}{{\bf X }}
\newcommand{\bI}{{\bf I }}

\newcommand{\bZ}{{\bf Z }}
\newcommand{\bz}{{\bf z }}
\newcommand{\bx}{{\mathbf{x}}}
\newcommand{\bM}{{\bf M}}

\newtheorem{Def}{Definition}
\newtheorem{Pro}{Proposition}

\newtheorem{Rem}{Remark}
\allowdisplaybreaks
\ifCLASSINFOpdf
\else
\fi

\begin{document}

\title{Distributed Fusion with Multi-Bernoulli Filter based on Generalized Covariance Intersection}

\author{Bailu Wang, Wei Yi*, Reza Hoseinnezhad,
Suqi Li, Lingjiang Kong, Xiaobo Yang
\thanks{This work was supported by the National Natural Science Foundation of China under Grants 61301266, the Chinese Postdoctoral Science Foundation under Grant 2014M550465, and supported by the ARC Discovery Project grant DP130104404. \emph{(Corresponding author: Wei Yi.)}

B. Wang, W. Yi, S. Li, L. Kong and X. Yang are with the School of  Electronic Engineering, University of Electronic Science and Technology of China, Chengdu 611731, China (Email: w\_b\_l3020@163.com; kussoyi@gmail.com;  qi\_qi\_zhu1210@163.com; lingjiang.kong@gmail.com; yangxb\_uestc@hotmail.com).

R. Hoseinnezhad is with the School of Aerospace, Mechanical and Manufacturing Engineering, RMIT University, Victoria 3083, Australia (Email: reza.hoseinnezhad@rmit.edu.au). 
}
}
\vspace{-0.2in}
\maketitle
 \thispagestyle{empty}
\begin{abstract}
In this paper, we propose a distributed multi-object tracking algorithm through the use of multi-Bernoulli (MB) filter based on generalized Covariance Intersection (G-CI).
Our analyses show that the G-CI fusion with two MB posterior distributions does not admit an accurate closed-form expression.
To solve this problem, we firstly approximate the fused posterior as the unlabeled version of $\delta$-generalized labeled multi-Bernoulli ($\delta$-GLMB) distribution, referred to as generalized multi-Bernoulli (GMB) distribution.
Then, to allow the subsequent fusion with another multi-Bernoulli posterior distribution, e.g., fusion with a third sensor node in the sensor network, or fusion in the feedback working mode, we further approximate the fused GMB posterior distribution as an MB distribution which matches its first-order statistical moment.
The proposed  fusion algorithm is implemented using sequential Monte Carlo technique and its performance is highlighted by numerical results.
\end{abstract}


%
\IEEEpeerreviewmaketitle
\section{Introduction}
Distributed multi-sensor multi-object tracking (DMMT) methods generally benefit from lower communication cost and higher fault tolerance, compared with centralized multi-object tracking solutions. As such, they have increasingly attracted interest from tracking community. Devising  DMMT solutions becomes particularly challenging when the correlations between the estimates from different sensors are not known.  The optimal solution to this problem was developed in \cite{CY-Chong}, but the computational cost of calculating the common information can make the solution intractable in many real-world applications. An alternative is  the suboptimal fusion technique, namely, Generalized Covariance Intersection (G-CI)  or  exponential mixture densities (EMD) \cite{EMD-Julier} proposed by Mahler \cite{Mahler-1}. G-CI\footnote{G-CI has been applied in the literature with various names such as Chernoff fusion~\cite{chernoff-fusion} and  geometric mean density (GMD) \cite{GMD-fusion}.} is the generalization of Covariance Intersection \cite{Uhlmann} which only utilizes the mean and covariance and is limited to Gaussian posteriors. The highlight of G-CI is that it is capable to fuse  both Gaussian \cite{EMD-Julier} \cite{chernoff-fusion}  and non-Gaussian  multi-object distributions from different sensors with completely unknown correlations.

 Following the work of Mahler~\cite{Mahler-1}, Clark \textit{et~al.}~\cite{Clark} derived several tractable formulations of G-CI fusion for  special types of multi-object distributions  including Poisson, independent identically distributed (i.i.d.) clusters and Bernoulli distributions. Using these formulations, a sequential Monte Carlo (SMC) realization  of the distributed fusion with probability hypothesis density (PHD) filter was presented in \cite{Uney-2}.
 Meanwhile, the problem of DMMT with a Gaussian mixture cardinalized PHD (GM-CPHD) filter was addressed in \cite{Battistelli}. Furthermore, the work of distributed detection and tracking  with Bernoulli
 filter over a Doppler-shift sensor network was completed in \cite{Mehmet}.

In addition to the PHD and CPHD filters \cite{MeMBer_Mahler,P-Brcac,PHD-Vo,PHD_TBD,Vo-CPHD}, the multi-Bernoulli (MB) filter is also a promising multi-object tracking algorithm in the framework of random finite sets (RFS). Compared with PHD and CPHD filters, the MB filter can be more efficient and accurate in problems that require particle implementations or target individual existence probabilities. The reason is that the MB filter \cite{MeMBer_Vo2}, \cite{MeMber_Vo3}  directly propagates the multi-object distribution, not its moments. Furthermore, it  does not require an additional process, such as a clustering step, to extract the multi-object state estimate. MB filters have  been successfully applied in a host of practical problems. Examples include
radar tracking \cite{Vo_radar_target_nonline}, image tracking \cite{MeMber_Vo3,Gunes_sonar}, ground target  tracking \cite{Reza_ground_target}, sensor control \cite{Reza_sensor_control_AES,Reza_sensor_control_letter}, audio and video data tracking \cite{Reza_audio_visual}, visual tracking and cell tracking \cite{Reza_visual_tracking}, and mobile multi-object tracking \cite{MeMber_Wei1,MeMber_Wei2}. Novel extensions of the MB recursion have also been proposed in \cite{Dunne_MM} for multiple models, and hybrid multi-Bernoulli and Poisson multi-target filters were also proposed in \cite{Williams}. To the best of our knowledge, the problem of DMMT with  MB filter considering the unknown level of correlation among sensors has not been well addressed.
The challenge lies in intractability of deriving a  closed-form expression of G-CI fusion with MB distributions.

Recently, the notion of labeled RFS was introduced to address target trajectories and their uniqueness in \cite{LMB_Vo, LMB_Vo2, delta_GLMB, LRFS_Bear_Vo, LRFS_Papi_Kim, Fantacci-BN, GLMB_Papi_Vo}. To investigate the DMTT of labeled densities, Fantacci~\textit{et al.}~\cite{Fantacci-BT} derived the closed-form
solutions of GCI fusion with marginalized $\delta$-generalized labeled multi-Bernoulli (M$\delta$-GLMB) and labeled multi-Bernoulli (LMB) densities
based on the assumption that different sensors share the
same  label space. ``Sharing the same
space'' demands that not only label spaces from different sensors are the same numerically, but also the same element from different sensors has the same physical implication, or indeed denotes the same object. This assumption is hardly valid in practice, which is also referred to as the ``label space mismatching'' phenomenon and is analyzed in detail in \cite{GCI-GMB,GCI-LSM}. Wang~\textit{et al.}~\cite{GCI-GMB, GCI-LSM} have recently suggested two promising solutions to cope with the ``label space mismatching'' phenomenon.

In some applications, the labels of the object are of great importance; still there are many cases where one might say that ``a threat is a threat'' and  we have no interest in which target is which. For example, in collision avoidance systems, the objective is not to distinguish the identities of  cars, but to avoid them regardless of their identities. In such cases, the labeled multi-object density is not required. Moreover, when there is uncertainty in labelling the targets, e.g., in presence of closely spaced targets, the labeled posterior will be  multi-modal and this may  affect the  performance \cite{Set_JPDA}.  Hence, the unlabeled  filters still remain in current and widespread use. MB filter is a kind of  unlabeled filter and its successful  applications spin over many areas as mentioned above.  Hence, it is significant to explore the generalization of MB filter to distributed environment.

In this paper, the problem of DMMT with MB filters over a sensor network is investigated. The major contributions are two-fold:
\begin{enumerate}[1)]
\item \textit{We propose a distributed fusion algorithm, namely, GCI fusion with MB filter (GCI-MB). } A tractable closed-form formulation of GCI fusion with MB posteriors are obtained via two reasonable approximations.
\item \textit{We implement the proposed fusion  algorithm using SMC technique.} The main challenge with SMC implementation of the proposed MB-fusion solution is that  neither the support nor the number of particles are guaranteed to be the same in different sensor nodes. To this end, a Kernel Density Estimation method \cite{Uney-2, KDE-1, KDE-2} is employed to convert the local particle sets to a Gaussian Mixture model (GMM) obtaining a continuous approximation.
\end{enumerate}
In numerical results,  the performance of the proposed fusion algorithm with SMC implementation is verified.

Preliminary results have been announced in the conference paper \cite{GCI-MB-Con}. This paper presents a more complete theoretical and numerical study. In Section II an overview of  multi-object tracking with RFS and G-CI fusion rule is given. Section III describes our approach for DMMT. We firstly derive  the  closed-form expression of the fused posterior by approximating it as generalized multi-Bernoulli (GMB) distribution. Then we approximate the fused GMB posterior with an MB distribution with matching first-order statistical moment. In Section IV, we present the SMC based implementation of the proposed distributed fusion algorithm. The performance of the proposed algorithm is analyzed in two distributed multi-object tracking scenarios in Section V. Then,
some concluding remarks are given in Section VI.
\section{Background}
\subsection{Notation}

To admit arbitrary arguments like sets, vectors and integers, the generalized Kronecker delta function is given by
\begin{equation}\label{delta}
  \delta_{\mathbf{Y}}(\bX)\triangleq\left\{\begin{array}{l}
\!\!1,\,\,\,\, \mbox{if $\bX = \mathbf Y$} \\
\!\!0,\,\,\,\, \mbox{otherwise}
\end{array}\right.
\end{equation}
and the inclusion function is given by
\begin{equation}\label{inclusion function}
  1_{\mathbf{Y}}(\bX)\triangleq\left\{\begin{array}{l}
\!\!1, \,\,\,\,\mbox{if $\bX \subseteq \mathbf{Y}$}\\
\!\!0. \,\,\,\,\mbox{otherwise}
\end{array}\right.
\end{equation}

The vector integrals  on $\mathbb{X}$ are using the  standard inner product
notation. For functions $a(\bx)$ and $b(\bx)$ defined on $\mathbb{X}$, the inner product notation is represented as
$\big<a,b\big>=\int_\mathbb{X} a(\bx)b(\bx)d\bx$.
\subsection{Multi-object Bayesian Filter}
Finite Set Statistics (FISST) proposed by Mahler, has provided a rigorous and elegant mathematical framework for the multi-object detection,
tracking and classification problems in a unified Bayesian paradigm.

In the FISST framework, the multi-object state at time $k$ is naturally represented as an RFS $\bX^k=\{\bx^k_1,\bx^k_2,\ldots,\bx^k_{n}\}\in \mathcal{F}(\mathbb{X})$,
where $\mathbb{X}=\mathbb{R}^\nu$ is the single object state space with the dimension $\nu$, $\mathcal{F}(\mathbb{X})$ is the space of finite subsets of $\mathbb{X}$.
Each single object state $\bx^k_i=[{\mathbf{p}^{k}_{i}}^{\top}\ \ {\mathbf{v}^k_i}^{\top}]^{\top} \in \mathbb{R}^\nu$ comprises the positions $\mathbf{p}^k_i\in \mathbb{R}^{\nu/2}$
and velocities $\mathbf{v}^k_i\in \mathbb{R}^{\nu/2}$, where ``$^\top$'' denotes the matrix transpose.

Let $\bZ^{k}$ denotes the observation at time $k$ and $\bZ^{1:k}=(\bZ^{1},\ldots,\bZ^{k})$ denotes the history of observation from time $1$ to
time $k$. The optimal multi-object Bayesian filter propagates RFS based posterior density $\pi(\bX^k|\bZ^{1:k})$ conditioned on $\bZ^{1:k}$
in time with the following recursion  \cite{MeMBer_Mahler}:
\begin{align}
\label{Optimal-prediction}\pi(\bX^k|\bZ^{1:k-1})\!&=\!\int\! f(\bX^k|\bX^{k-1})\pi(\bX^{k-1}|\bZ^{1:k-1})\delta \bX^{k-1}\\
\label{Optimal-update}\pi(\bX^k|\bZ^{1:k})\!&=\!\frac{g(\bZ^k|\bX^k)\pi(\bX^k|\bZ^{1:k-1})}{\int g(\bZ^k|\bX^k)\pi(\bX^k|\bZ^{1:k-1})\delta \bX^{k}}
\end{align}
where $f(\bX^k|\bX^{k-1})$ is the multi-object Markov transition function,  $g(\bZ^k|\bX^k)$ is the multi-object likelihood function, and set integral is defined by~\cite{MeMBer_Mahler}
\begin{equation}\label{set integral}
  \int\! f(\bX)\delta \bX\!=\sum_{n=0}^\infty \frac{1}{n!}\int\! f(\{\bx_1,\cdots,\bx_n\})d\bx_1\cdots d\bx_n.
\end{equation}
\subsection{Multi-Bernoulli Distribution}
A random set $\textbf{X}$ with multi-Bernoulli (MB) distribution is defined as the union of $M$ independent Bernoulli random sets~$\bX^{(\ell)}$~\cite{Mahler-1},
\begin{equation}
\bX=\bigcup_{\ell=1}^{M}\bX^{(\ell)}.
\end{equation}
The MB distribution is completely characterized by a set of parameters $\{(r^{(\ell)},p^{(\ell)})\}_{\ell=1}^{M}$, where $r^{(\ell)}$ denotes the existence probability and $p^{(\ell)}(\cdot)$ denotes the probability density of the $\ell$-th Bernoulli random set. The multi-object probability density of an MB RFS is given by \cite{Mahler-1},
\begin{align}\label{multi-Bernoulli}
\begin{split}
&\pi(\left\{\mathbf{x}_1,\ldots,\mathbf{x}_n\right\})\\
&\,\,\,\,\,=\sum_{1\leq i^1\neq\ldots\neq i^n\leq M}Q^{(i^1,\cdots,i^n)}\prod_{j=1}^n p^{(i^{j})}(\bx_{j})
\end{split}
\end{align}
where
\begin{align}\label{Phi}
\begin{split}
Q^{(i^1,\cdots,i^n)}=\prod_{\ell=1}^{M}(1-r^{(\ell)})\prod_{j=1}^n\frac{r^{(i^{j})}}{1-r^{(i^{j})}}.
\end{split}
\end{align}

\subsection{$\delta$-Generalized Labeled Multi Bernoulli Distribution}
A $\delta$-Generalized Labeled Multi Bernoulli~($\delta$-GLMB) distribution is defined for labeled RFSs. It is parametrized as follows~\cite{LMB_Vo2}:
\begin{align}\label{delta-GLMB}
\begin{split}
\bpi(\bX)=\triangle(\bX)\sum_{(I,\xi)\in\mathcal{F}(\mathbb{L})\times\Xi}\omega^{(\xi)}(I)\ \delta_I(\mathcal{L}(\bX))[p^{(\xi)}]^\bX
\end{split}
\end{align}
where $\mathbb{L}$ is a discrete label space, $\Xi$ is a discrete space, $\xi$ denotes a point in the space $\Xi$, the factor $\triangle(\bX) = \delta_{|\bX|}(\mathcal{L}(\bX))$ is included to guarantee unique labels, each $p^{(\xi)}(\bx)$ is a probability density over the joint space of single-object states and labels (thus, also denoted by $p^{(\xi)}(x,\ell)$ or by $p^{(\ell,\xi)}(x)$, and each $\omega^{(\xi)}(I)$ is a non-negative weight. The weights are normalized,
\begin{equation}\label{omega}
  \sum_{(I,\xi)\in\mathcal{F}(\mathbb{L})\times\Xi}\omega^{(\xi)}(I)=1.
\end{equation}

An unlabeled $\delta$-GLMB distribution has the following general form~\cite{GCI-GMB}:
\begin{equation}\label{delta-GMB}
\pi(\{x_1,\ldots,x_n\})=\sum_{\sigma}\sum_{(I,\xi)\in\mathcal{F}_n(\mathbb{L})\times\Xi}\omega^{(I,\xi)}\prod_{i=1}^{n}p^{(\bI^v(i),\xi)}(x_{\sigma(i)})
\end{equation}
where $\mathcal{F}_n(\mathbb{L})$ is the space of finite subsets of $\mathbb{L}$ with cardinality  $n$, $\bII^v\in\mathbb{N}^{|I|}$ denotes the vector constructed by stacking the elements of $I$ in some sorted order,  $\sigma$ denotes a permutation of $\{1,\cdots,n\}$ and $\sum_{\sigma}$ denotes the sum over all such permutations. In this paper, we refer  to the unlabeled $\delta$-GLMB distribution as Generalized Multi-Bernoulli (GMB) distribution.

\subsection{Distributed Data Fusion}

Consider two nodes 1 and 2 in a sensor network. At time $k$, the nodes maintain their local posteriors $\pi_{1}(\bX^k|\bZ_{1}^{1:k})$ and $\pi_{2}(\bX^k|\bZ_{2}^{1:k})$ which are both labeled RFS multi-object densities. Node 1 transmits its posterior to node 2 where it is to be fused with node 2 local posterior to obtain a joint posterior denoted by
\begin{equation}\label{fused posterior}
  \pi_\omega(\bX^{k}|\bZ^{1:k}_1,\bZ^{1:k}_2)=  \pi_\omega(\bX^{k}|\bZ^{1:k}_1\cup\bZ^{1:k}_2)
\end{equation}
 where $\pi_\omega(\mathbf X^{k}|\mathbf Z^k_1,\mathbf Z^k_2)$ denotes the fused posterior of distributed fusion.
 It is important to note that common process noise arises whenever both nodes track the same target and common observation noise arises after the nodes exchange their local estimates with one another. Thus, in practical applications, $\pi_{1}(\bX^k|\bZ_{1}^{1:k})$ and $\pi_{2}(\bX^k|\bZ_{2}^{1:k})$ are not distribution of independent variables. Considering the unknown level of correlation among nodes, the following solution to the fusion problem was developed by Chong, Mori and Chang \cite{CY-Chong},
\begin{align}\label{optimal}
\begin{split}
\pi_{\omega}(\bX^k|\bZ_{1}^{1:k},\bZ_{2}^{1:k})\propto\frac{\pi_{1}(\bX^k|\bZ_{1}^{1:k})\pi_{2}(\bX^k|\bZ_{2}^{1:k})}
                                              {\pi(\bX^k|\bZ_{1}^{1:k}\cap \bZ_{2}^{1:k})}.
\end{split}
\end{align}
In many applications, the computation of posterior given common information between sensors, $\pi(\bx^k|\bZ_{1}^{1:k}\cap \bZ_{2}^{1:k})$, is not straightforward, and the above fusion rule cannot be easily implemented. To overcome this issue, the G-CI fusion rule, which specifically extends FISST to distributed environments, has been proposed by Mahler \cite{Mahler-1}. Under this generalization, the fused posterior is the geometric mean, or the exponential mixture of the local posteriors,
\begin{align}\label{G-CI}
\begin{split}
\!\!\!\pi_{\omega}(\bX^k|\bZ_1^{1:k},\bZ_{2}^{1:k})\!=\!\frac{\pi_{1}(\bX^k|\bZ_{1}^{1:k})^{\omega_1}\pi_{2}(\bX^k|\bZ_{2}^{1:k})^{\omega_2}}
                                              {\int \pi_{1}(\bX^k|\bZ_{1}^{1:k})^{\omega_1}\pi_{2}(\bX^k|\bZ_{2}^{1})^{\omega_2}\delta \bX}
\end{split}
\end{align}
where $\omega_1$, $\omega_2$ ($\omega_1+\omega_2=1$) are the parameters determining the relative fusion weight of each nodes.

The fused posterior given by equation (\ref{G-CI}) minimizes the weighted sum of its Kullback-Leibler divergence (KLD) \cite{Battistelli} with respect to  two given distributions,
\begin{equation}\label{EMD}
\begin{split}
  \pi_\omega=\arg \min_\pi(\omega_1D_{\emph{KL}}(\pi\parallel \pi_1)+\omega_2 D_{\emph{KL}}(\pi\parallel \pi_2))
  \end{split}
\end{equation}
where $D_{\emph{KL}}$ denotes the KLD defined as
\begin{equation}\label{KLD}
\begin{split}
 D_{\emph{KL}}(f||g)\triangleq \int f(\bX)\log{\frac{f(\bX)}{g(\bX)}}\delta\bx
  \end{split}
\end{equation}
where the integral in (\ref{KLD}) is generally a set integral. For convenience of notations, in what follows we omit explicit references to the time index $k$.

\section{Distributed Fusion with MB Filters }
In this section, we present a tractable closed-form solution for G-CI based distributed fusion of multi-Bernoulli posteriors that are locally formed in separate nodes of a sensor network. Each local sensor performs MB filtering  and outputs a MB posterior in the form of (\ref{multi-Bernoulli}). Depending on the type of the local measurement acquitted by the sensor node, the local MB filter may use various observation models such as point observation model~\cite{MeMBer_Vo2} or image observation model~\cite{MeMber_Vo3}. Through a practical approximation, we show that fusion of two MB posteriors using G-CI formula~(\ref{G-CI}) leads to a GMB-type multi-object density whose parameters can be directly calculated in terms of the two MB distribution parameters. We then approximate the fused GMB distribution with an MB distribution that has the same first moment. This can be fed back to the sensor network nodes for the next iteration of local MB filtering.

\subsection{G-CI Fusion}
When fusing MB distributions based on the G-CI fusion formula~(\ref{G-CI}), the main challenge is that for each MB distribution $\pi(\bX)$, the term  $\pi(\bX)^\omega$ has a form of fractional order exponential power of a sum, ${\left(\sum_{i=1}^n d_i \right)}^\omega$,  which is computationally intractable. Its value could be approximated  using numerical solutions, such as grid based approximation. However, this approach suffers from the curse of dimensionality and is prohibitively expensive in general. Therefore, a feasible and practical approximation of ${\pi(\bX)}^\omega$  is required.

In \cite{Battistelli,July}, the following approximation has been introduced to calculate  ${\pi(\bX)}^{\omega}$ where $\pi(\bX)$ is a single-object distribution formulated as a mixture of \textit{well separated} Gaussian components:
\begin{align}\label{generalized_approximation}
\begin{split}
\left(\small\sum_{i} d_i\right)^\omega\approx\small\sum_i d_i^\omega.
\end{split}
\end{align}
In the following, we derive a similar approximation for a multi-object distribution  ${\pi(\bX)}$ that is formulated as an MB distribution which is the union of \textit{well separated} Bernoulli components. Our derivation also clarifies what being \textit{well separated} means for Bernoulli components of the MB distribution.

In order to make the derivations presented in this section more compact, we represent the MB distribution in (\ref{multi-Bernoulli}) in another form.
For each cardinality $n \leq M$, we denote the ensemble of all possible ordered combinations of $n$ distinct indices between 1 and $M$, by the summing joint-index space $H(n)$,
\begin{equation}\label{P_MN}
H(n)=\{(i^1,\cdots,i^n)\in\mathbb{N}^n |1\leq i^1\neq \cdots\neq i^n\leq M\},
\end{equation}
where $\mathbb{N}$ is referred to as the set of all natural numbers.

Using this notation, the MB distribution (\ref{multi-Bernoulli}) can be rewritten as
\begin{equation}\label{multi-Bernoulli_V3}
\pi(\{\bx_1,\cdots,\bx_n\})=\sum_{\bII_h\in H(n)}Q^{\bII_h}\prod_{i=1}^n p^{(\bII_h(i))}(\bx_i).
\end{equation}
Therefore,
\begin{equation}\label{term_w}
\begin{split}
\pi(\{\bx_1,\cdots,\bx_n\})^\omega={\left(\sum_{\bII_h\in H(n)}Q^{\bII_h}\prod_{i=1}^n p^{(\bII_h(i))}(\bx_i)\right)}^\omega.
\end{split}
\end{equation}

We will show that is the Bernoulli components of the MB distribution are well-separated, the powered sum presented in the above equation can be approximated by the sum of powers. Firstly, we introduce the concept of highest posterior density (HPD) region \cite{HPD}, which is important for the derivation that follows.
\begin{Def}
\normalfont{
Let $p(X|{Z})$ be a posterior density function. A region $R$ in the space of $X$ is called an HPD of confidence $\lambda$ if
\begin{enumerate}[a)]
\item $\Pr\{X\in R|{Z}\}=\lambda$;
\item for $X_1\in R$ and $X_2\notin R$,
\begin{equation}\label{HPD}
\begin{split}
p(X_1|{Z})\geq p(X_2|{Z}).
\end{split}
\end{equation}
\end{enumerate}
}
\end{Def}
 The posterior density for every point inside the HPD region is greater than that for every point outside of region. Thus, the region includes the more probable values of $X$.  Usually, the confidence $\lambda$ is set to be very close to one, e.g. $\lambda=0.90$. Thus, $p(X|Z)$ is negligible for $X\notin R$ and can be approximated~with~0.

\begin{Def}
\normalfont
Consider an MB posterior $\pi=\left\{\left(r^{(\ell)},p^{(\ell)}(\cdot)\right)\right\}_{\ell=1}^M$. If $\mathbb{X}_\ell$ is the HPD of confidence $\lambda$ for $p^{(\ell)}(\cdot)$, then the Bernoulli components of $\pi(\bX)$ are said to be  \underline{mutually $\lambda\times 100\%$  separated} if,
$$
\forall\ell\neq\ell^{\prime},  \ \ \mathbb{X}_{\ell}\cap\mathbb{X}_{\ell^{^{\prime}}}=\emptyset.
$$
\end{Def}

\begin{Rem}
\normalfont{
In practical multi-object tracking scenarios, the HPD of posterior $p^{(\ell)}(\bx)$ is influenced by many factors, e.g., the true target states, the maneuverability and signal-to-noise (SNR) of targets.  Usually the true single target state corresponding to each Bernoulli component determines the center of its HPD region. Furthermore, the width of HPD region of a Bernoulli component is smaller with lower maneuverability and higher SNR. In such practical scenarios, the MB distributions propagated through an MB filter (and G-CI fused density in sensor network applications) can be easily assumed to be mutually separated with very high confidence.}
\end{Rem}

\begin{Rem}
\normalfont{
In common SMC implementations of the MB filter, there is a merging step after update, in which the Bernoulli components whose means are too close to each other are merged into one Bernoulli component. Thus, we can practically assume that the posteriors that are to be fused in the sensor network are always well-separated ($\lambda$ is very close to one).
}
\end{Rem}

\begin{Pro}\label{P1}
Assume that the Bernoulli components of an MB posterior density, denoted by $\pi=\left\{\left(r^{(\ell)},p^{(\ell)}(\cdot)\right)\right\}_{\ell=1}^M$, are mutually $\lambda\times 100\%$ separated. Denote the HPD of confidence $\lambda$ for $p^{(\ell)}(\bx)$ by $\mathbb{X}_{\ell}$. For an indexing sequence $\bI_h \in H(n)$, consider the multi-variate posterior $\mathfrak{p}(\bx_{1:n};\bI_h) = \prod_{i=1}^{n}p^{(\bI_h(i))}(\bx_i)$ where $\bx_{1:n}$ denotes $(\bx_1,\ldots,\bx_n)$. If the confidence level $\lambda$ is close to one, then the HPD of confidence $\lambda^n$ for $\mathfrak{p}(\bx_{1:n};\bI_h)$ can be approximated with
$$
\mathbb{X}_{\bI_h} \approx \widehat{\mathbb{X}}_{\bII_h} = \mathbb{X}_{\bI_h(1)} \times \mathbb{X}_{\bI_h(2)} \times \cdots \times \mathbb{X}_{\bI_h(n)}.
$$
\end{Pro}

\begin{proof}
The probability associated with the above HPD is given by:
\begin{eqnarray}
\nonumber
\Pr(\bx_{1:n}\in \widehat{\mathbb{X}}_{\bII_h}) & = & \int_{\widehat{\mathbb{X}}_{\bII_h}} \mathfrak{p}(\bx_{1:n};\bI_h) d\bx_{1:n} \\
\nonumber
& = & \int_{\widehat{\mathbb{X}}_{\bII_h}} \prod_{i=1}^{n}p^{(\bI_h(i))}(\bx_i) d\bx_i\\
\nonumber
& = & \prod_{i=1}^{n} \int_{\mathbb{X}_{\bI_h(i)}} p^{(\bI_h(i))}(\bx_i) d\bx_i\\
\nonumber
& = & \prod_{i=1}^{n} \lambda\\
& = & \lambda^n.
\end{eqnarray}
Furthermore, consider two $n$-tuples $\bx_{1:n} \in \widehat{\mathbb{X}}_{\bII_h}$ and $\mathbf{y}_{1:n} \notin \widehat{\mathbb{X}}_{\bII_h}$. We argue that the condition $\mathfrak{p}(\bx_{1:n};\bI_h) \geqslant \mathfrak{p}(\mathbf{y}_{1:n};\bI_h)$ holds for almost all possible pairs of $\bx_{1:n},\mathbf{y}_{1:n}$. Without loss of generality, let us assume that the first $\mathfrak{n} \leqslant n$ elements of $\mathbf{y}_{1:n}$ are not in their correspondent HPD regions and the rest are. The inverse condition  $\mathfrak{p}(\bx_{1:n};\bI_h) < \mathfrak{p}(\mathbf{y}_{1:n};\bI_h)$ can be rewritten as
$$
\prod_{i=1}^{n}p^{(\bI_h(i))}(\bx_i) < \prod_{i=1}^{n}p^{(\bI_h(i))}(\mathbf{y}_i).
$$
We note that for $i = 1, \ldots, \mathfrak{n}$, $p^{(\bI_h(i))}(\bx_i) \geqslant p^{(\bI_h(i))}(\mathbf{y}_i)$. Thus, in order for the above inverse condition to hold, the rest of the elements of $\mathbf{y}_{1:n}$ must associate with densities much larger than the ones at $\bx_{1:n}$ in such a way that when multiplied by the first $\mathfrak{n}$ densities, the product still becomes larger than the product of densities associated with elements of $\bx_{1:n}$. We argue that when the confidence level $\lambda$ is large, the bulk mass of distribution is covered by the HPD and the densities associated with values outside the HPD are expected to be negligible. Thus, the product $\prod_{i=1}^{\mathfrak{n}} p^{(\bI_h(i))}(\mathbf{y}_i)$ is expected to be so small that its product with the rest of the terms can rarely become large enough to exceed the total product of densities at $\bx_i$'s, $\prod_{i=1}^{n}p^{(\bI_h(i))}(\bx_i)$. More precisely, the above inverse condition can rarely be held, and
$$
\mathfrak{p}(\bx_{1:n};\bI_h) \geqslant \mathfrak{p}(\mathbf{y}_{1:n};\bI_h)
$$
holds for almost all possible pairs of $\bx_{1:n} \in \widehat{\mathbb{X}}_{\bII_h}$ and $\mathbf{y}_{1:n} \notin \widehat{\mathbb{X}}_{\bII_h}$.
\end{proof}

\begin{Pro}\label{P2}
If the Bernoulli components of an MB posterior density denoted by $\pi=\left\{\left(r^{(\ell)},p^{(\ell)}(\cdot)\right)\right\}_{\ell=1}^M$ are mutually $\lambda\times 100\%$ separated and $\lambda$ is very close to $1$ (e.g. $\lambda \geqslant 0.9$), then
\begin{equation}\label{fuse-2}
\begin{split}
{{\pi}(\!\left\{\bx_1\!,\!\ldots\!,\!\bx_n\right\}\!)}^{\omega}
\approx\sum_{\bII_h\in H(n)}\left({Q^{\bII_h}}\right)^{\omega}\left(\prod_{i=1}^{n}p^{(\bII_h(i))}(\bx_i)\right)^{\omega}.
\end{split}
\end{equation}
\end{Pro}

\begin{proof}
From proposition~\ref{P1}, if $\mathbb{X}_{\ell}$  is the HPD region of $p^{(\ell)}(\bx)$ with confidence $\lambda$, for each cardinality $n$, the HPD region $\mathbb{X}_{\bII_h}$ of each product term $\prod_{i=1}^n p^{(\bII_h(i))}(\bx_i), \bII_h\in H(n)$ with confidence $\lambda^n$  can be approximately represented as
\begin{equation}
 {\mathbb{X}}_{\bII_h} \approx \widehat{\mathbb{X}}_{\bII_h} = \mathbb{X}_{\bII_h(1)}\times\cdots\times\mathbb{X}_{\bII_h(n)}.
\end{equation}

Note that since the $\bx_n$s are  mutually orthogonal, it is reasonable and convenient to use the $ \widehat{\mathbb{X}}_{\bII_h}$ whose geometric shape is rule,  to approximate the true $\mathbb{X}_{\bII_h}$.

Given the  applicable conditions  that  $\forall\ell\neq\ell^{\prime},  1\leq \ell,\ell' \leq M$, $\mathbb{X}_{\ell}\cap\mathbb{X}_{\ell^{^{\prime}}}=\emptyset$,  we have
\begin{equation}\label{intersection_1}
\begin{split}
\widehat{\mathbb{X}}_{\bII_h}\cap\widehat{\mathbb{X}}_{\bII_h'}=\emptyset, \,\,\forall \bII_h\neq\bII_h',\,\,\bII_h,\bII_h'\in H(n).
\end{split}
\end{equation}

If the single-object state space is denoted by $\mathbb{X}$, the multi-object state space with cardinality $n$ will be $\mathbb{X}^n$. For any $\bx_{1:n}\in \mathbb{X}^n$, we consider two possible cases:
\begin{itemize}
\item[-] If for some $\bI_h$, $\bx_{1:n}\in \widehat{\mathbb{X}}_{\bI_h}$, then from~(\ref{intersection_1}), it cannot be in any other HPD region $\widehat{\mathbb{X}}_{\bI_{h'}},~\bI_{h'} \neq \bI_h.$ Thus, among the product terms $\prod_{i=1}^n p^{(\bII_h(i))}(\bx_i)$ that appear in the sum of RHS of equation~(\ref{multi-Bernoulli_V3}), only one of them will be dominant and the others will have negligible values, i.e.
\begin{equation}\label{app_1}
\begin{split}
&{\pi(\!\left\{\bx_1\!,\!\ldots\!,\!\bx_n\right\}\!)}
\approx Q^{\bII_h}\prod_{i=1}^n p^{(\bII_h(i))}(\bx_i)
\end{split}
\end{equation}
and therefore,
\begin{equation}\label{approx-0}
\begin{split}
&{{\pi}(\!\left\{\bx_1\!,\!\ldots\!,\!\bx_n\right\}\!)}^{\omega}\approx
\left({Q^{\bII_h}}\right)^{\omega}\left(\prod_{i=1}^{n}p^{(\bII_h(i))}(\bx_i)\right)^{\omega}.
\end{split}
\end{equation}

It is important to note that the term $Q^{\bII_h}$ is the probability of joint existence of targets with labels $\bII_h$, thus $0 \leqslant Q^{\bII_h} \leqslant 1.$ This probability term $Q^{\bII_h}$ itself can be smaller than some of the probability terms for other labels, i.e. for some $\bI_{h'}$, we may have $Q^{\bII_h} < Q^{\bII_{h'}}$. However, for those other terms, the product of densities would be so small that $Q^{\bII_h}\prod_{i=1}^n p^{(\bII_h(i))}(\bx_i)$ would be still much larger than other terms with other indices $\bII_{h'}$.
\item[-] If the multi-object state value $\bx_{1:n}$ is in none of the HPD spaces $\{\widehat{\mathbb{X}}_{\bI_{h}}\}_{\bI_h \in H(n)}$, then all the product terms appearing in the sum of RHS of equation~(\ref{multi-Bernoulli_V3}) will be negligible, and the multi-object density at $\bx_{1:n}$ will be very close to zero. Accuracy of approximation of multi-object density is not of interest in such locations in the multi-object state space $\mathbb{X}^n$.
\end{itemize}

For an arbitrary multi-object state value $\bx_{1:n} = (\bx_1,\ldots,\bx_n)$, equation~(\ref{approx-0}) can be generalized to
\begin{equation}\label{approx-1}
\begin{split}
&{{\pi}(\!\left\{\bx_1\!,\!\ldots\!,\!\bx_n\right\}\!)}^{\omega}\approx
 \sum_{\bII_h\in H(n)}\left({Q^{\bII_h}}\right)^{\omega}\left(\prod_{i=1}^{n}p^{(\bII_h(i))}(\bx_i)\right)^{\omega}
\end{split}
\end{equation}
in which only one term from sum is dominant, depending on which HPD region $\bx_{1:n}$ belongs to.
\end{proof}

To show a numerical example and demonstrate the intuition behind this approximation, let us consider an  MB distribution
with three Bernoulli components  with probabilities of existence $r_1^{(1)}=0.8$,  $r_1^{(2)}=0.9$ and $r_1^{(3)}=0.9$, and densities $p_1^{(1)}(x)\sim\mathcal{N}(x;3,0.2)$, $p_1^{(2)}(x)\sim\mathcal{N}(x;4,0.2)$ and $p_1^{(3)}(x)\sim\mathcal{N}(x;7,0.2)$, with $x\in \mathbb{R}$. Since the densities are Gaussians characterized by their mean and covariance, the necessary condition of the MB posterior being well-separated is reduced to the means $\int \mathbf{x} p^{(\ell)}(\mathbf{x}) d\mathbf{x}$ and $\int \mathbf{x} p^{(\ell')}(\mathbf{x}) d\mathbf{x}$ for any $\ell\neq\ell'$, being well-separated as measured by their respective covariances.  Fig.~\ref{fig_sim} shows numerical values of the product terms for two hypotheses, one with cardinality $n=1$, and one with cardinality $n=2$. The figure clearly exemplifies how one product terms can significantly dominate the others, validating the accuracy of approximation~(\ref{approx-1}).

\begin{figure}[h]
\centering
\includegraphics[width=3.5in]{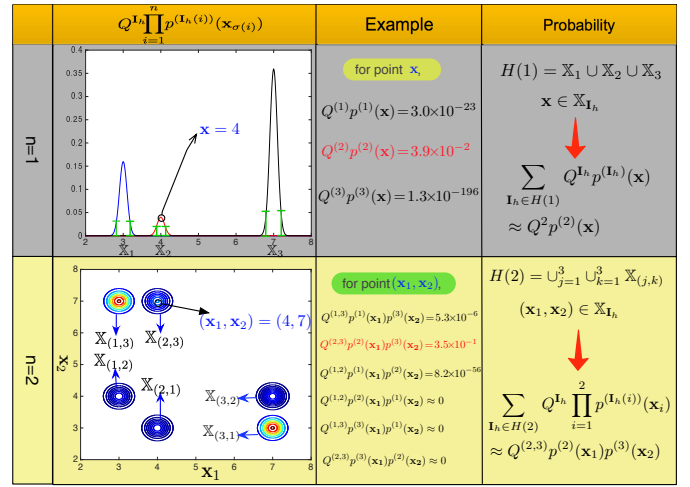} 
\caption{Example of an MB distribution with three components, and numerical values of the product terms for an $n=1$ dimensional hypothesis $\bX = \{4\}$ and an $n=2$ dimensional hypothesis $\bx = \{4,7\}$. The results show how in each case, regardless of hypothesized dimension, one product term in the sum formed by MB density in~(\ref{multi-Bernoulli_V3}) becomes much larger than others.}
\label{fig_sim}
\end{figure}

We note that the validity of the approximation in~(\ref{multi-Bernoulli_V3}) is not limited to Gaussian models. In the performance assessment section, the approximation is applied to the distributed multi-object tracking scenarios in which the multi-target posterior is not necessarily Gaussian distribution, and the results verify the validity and rationality of the approximation.

\subsection{The GMB Fused Distribution}
In this section, we present the MB distribution in a third form and define a fusion map describing the relationship between track outputs of two MB filters operating at two sensor nodes,  in order to explore the intuitionistic  mathematical structure of the fused distributions.

In addition to (\ref{multi-Bernoulli}) and (\ref{multi-Bernoulli_V3}), the third form of MB distribution could be expressed as
\begin{equation}\label{multi-Bernoulli_V2}
\begin{split}
\pi&(\!\left\{\bx_1\!,\!\ldots\!,\!\bx_n\right\}\!)
=\sum_{\sigma}\sum_{I\in \mathcal{F}_n(\mathbb{L})}Q^{I}\prod_{i=1}^{n}p^{(\bII^v(i))}(\bx_{\sigma(i)})
\end{split}
\end{equation}
where
 \begin{align}\label{Phi2_1}
\begin{split}
&Q^{I}=\prod_{\ell'\in I}{r^{(\ell')}}\prod_{\ell\in \mathbb{L}/I}(1-r^{(\ell)})
\end{split}
\end{align}
and $\mathbb{L}\!\triangleq\!\{1,\ldots,M\}$ is the index set of the MB distribution.

Consider two posteriors output by two sensors $s=1,\,2$, parametrized by $\pi_s=\{(r_s^{(\ell)},p_s^{(\ell)})\}_{\ell\in\mathbb{L}_{s}},\,\,s=1,\,2$, with $\mathbb{L}_s=\{1,\cdots,M_s\}$. Omitting the conditioning on the observations for convenience, we represent $\pi_s$ in the form of (\ref{MB-sensor}) as
\begin{align}\label{MB-sensor}
\pi_s=&\sum_{\sigma_s}\sum_{I_s\in \mathcal{F}_n(\mathbb{L}_s)}Q^{I_s}\prod_{i=1}^{n}p^{(\bII_s^v(i))}(\bx_{\sigma(i)}),\,\,\,s=1,\,2.
\end{align}

\begin{Def}
Without loss of generality, assume that  $|\mathbb{L}_1|\leq|\mathbb{L}_2|$. A fusion map is a function $\theta: I\in\!\!\mathcal{F}(\mathbb{L}_1)\!\rightarrow \!\mathbb{L}_2$ such that  $\theta(\ell)\!=\!\theta(\ell^{\ast})\!\!>\!\!0$
implies $\ell\!=\!\ell^{\ast}$. The set of all such fusion maps  is called fusion map space of $I$ denoted by $\Theta_{I}$, and the number of all fusion maps of $I$ is $A_{|I|}^{M_2}$, where $A_{N}^{M}$ denotes $N$-permutations of $M$. For notation convenience, we define $ \theta(I)\triangleq\{\theta(\ell), \ell\in I\}$.
\end{Def}
\begin{figure}[H]
\centering
\includegraphics[width=3.5in]{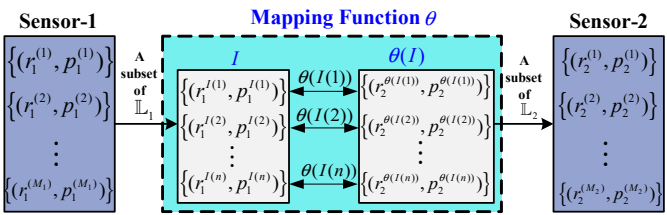} 
\caption{The sketches of the fusion map defined in Definition 1.  For any subset $I\in \mathcal{F}(\mathbb{L}_1)$,
         there is a subset $\theta(I) \in \mathcal{F}(\mathbb{L}_2)$ whose elements are one-to-one matching with the elements of $I$.}
\label{fig_fusion_map}
\end{figure}
\begin{Rem}
 \normalfont{
 Each fusion map denotes a hypothesis that a set of tracks in sensor 2 are one-to-one matching with
a set of tracks  in sensor 1 in the sense that the  matched tracks belong to the same targets, which is shown as in Fig.~\ref{fig_fusion_map}. The fusion map plays a similar role to the measurement-track association map in $\delta$-GLMB filter \cite{delta_GLMB}. For instance, consider two sensors, and their posteriors are $\{(r_1^{(\ell)},p_1^{(\ell)})\}_{\ell\in\mathbb{L}_1}$ and $\{(r_2^{(\ell)},p_2^{(\ell)})\}_{\ell\in\mathbb{L}_2}$, respectively, where $\mathbb{L}_1=\{1,2\}$ and $\mathbb{L}_2=\{1,2\}$. According to the Definition 3, there exist six fusion maps which are shown as in Fig.~\ref{fig_fusion_map_example}.
}
\end{Rem}

\begin{figure}[htbp]
\centering
\includegraphics[width=2.8in]{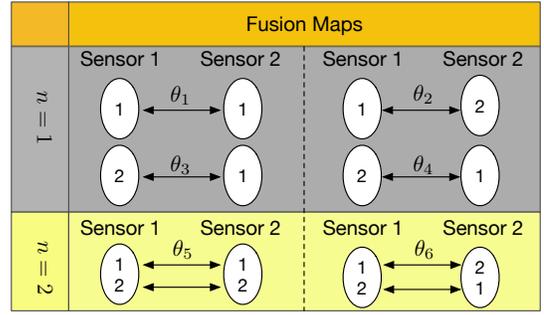}
\caption{\small{ An example of fusion maps.}}
\label{fig_fusion_map_example}
\end{figure}


\begin{Pro}\label{P2}
The EMD $\pi_\omega(\bX)$ of two MB distributions in (\ref{MB-sensor}) can be approximated as a GMB distribution of the form
\begin{align}\label{fuse-1}
\begin{split}
&\widetilde{\pi}_\omega(\left\{\bx_1,\ldots,\bx_n\right\})=\\&\sum_{\sigma}\sum_{(I_1,\theta)\in \mathcal{F}_n(\mathbb{L}_1) \times\Theta_{I_1}}w_\omega^{(I_1,\theta)}\prod_{i=1}^{n}p_\omega^{(\bII_1^v(i),\theta)}(\bx_{\sigma(i)})
\end{split}
\end{align}
where
\begin{align}
\label{fuse-w}
w_\omega^{(I_1,\theta)}&=\widetilde{w}_\omega^{(I_1,\theta)}\bigg{/} C\\
\label{fuse-p_p}p_{\omega}^{(\ell,\theta)}(\bx)&=\frac{p_1^{(\ell)}(\bx)^{\omega_1}p_2^{(\theta(\ell))}(\bx)^{\omega_2}}{Z_\omega^{(\ell,\theta)}},\,\,\,\ell\in I_1, \theta\in\Theta_{I_1}
\end{align}
with
\begin{align}
\label{Z_w} Z_\omega^{(\ell,\theta)}&=\int p_1^{(\ell)}(\mathbf{x})^{\omega_1}p_2^{(\theta(\ell))}(\mathbf{x})^{\omega_2}d\mathbf{x}\\
\label{fuse-w_p}\widetilde{w}_\omega^{(I_1,\theta)}&=\left(Q_1^{I_1}\right)^{\omega_1}\left(Q_2^{\theta(I_1)}\right)^{\omega_2}\prod_{\ell\in I_1}Z_\omega^{(\ell,\theta)}\\
\label{K-1}C&=\sum_{I_1\in \mathcal{F}(\mathbb{L}_1)}\sum_{\theta\in\Theta_{I_1}}\widetilde{w}_\omega^{(I_1,\theta)}.
\end{align}
\end{Pro}
\begin{proof}
 Firstly, applying (\ref{fuse-2}) to the MB distribution of the form (\ref{multi-Bernoulli_V2}), we can obtain
 \begin{equation}\label{approx-1-v2}
 \pi_s(\bX)^\omega\approx\sum_{\sigma_s}\sum_{I_s\in \mathcal{F}_n(\mathbb{L}_s)}{\left(Q^{I_s}\right)}^\omega{\left(\prod_{i=1}^{n}p_s^{(\bII_s^v(i))}(\bx_{i})\right)}^\omega.
 \end{equation}

 \begin{figure*}[ht]
\begin{equation}\label{eq:1}
\begin{split}
&\,\,\,\,\,\,\widetilde{\pi}_\omega(\left\{\mathbf{x}_1,\ldots,\mathbf{x}_n\right\})\\&
=\!\sum_{\sigma_1}\sum_{I_1\in \mathcal{F}_n(\mathbb{L}_1)}\sum_{\sigma_2}\sum_{I_2\in \mathcal{F}_n(\mathbb{L}_2)}\left(Q_1^{I_1}\prod_{i=1}^{n}p_{1}^{(\bII^v_1(i))}(\bx_{\sigma_1(i)})\right)^{\omega_1}\left(Q_2^{I_2}\prod_{i=1}^{n}p_{2}^{(\bII^v_2(i))}(\bx_{\sigma_2(i)})\right)^{\omega_2}\\&
=\!\sum_{\sigma_1}\sum_{I_1\in \mathcal{F}_n(\mathbb{L}_1)}\sum_{\sigma_2}\sum_{I_2\in \mathcal{F}_n(\mathbb{L}_2)}\left(Q_1^{I_1}\right)^{\omega_1}\left(Q_2^{I_2}\right)^{\omega_2}
\prod_{i=1}^n \left(p_1^{(\bII^v_1(i))}(\mathbf{x}_{\sigma_1(i)})\right)^{\omega_1}\left(p_2^{(\bII^v_2(i))}(\mathbf{x}_{\sigma_2(i)})\right)^{\omega_2}\\&
=\sum_{\sigma}\sum_{I_1\in \mathcal{F}_n(\mathbb{L}_1)}\sum_{\theta\in \Theta_{I_1}}\left(Q_1^{I_1}\right)^{\omega_1}\left(Q_2^{I_2}\right)^{\omega_2}
\prod_{i=1}^n \int \left(p_1^{(\bII^v_1(i))}(\mathbf{x}_{\sigma(i)})\right)^{\omega_1}\left(p_2^{(\theta(\bII^v_1(i)))}(\mathbf{x}_{\sigma(i)})\right)^{\omega_2} d\bx_{\sigma{(i)}} \\&
\,\,\,\,\,\,\,\,\,\,\,\,\,\,\,\,\,\,\,\,\,\,\,\,\,\,\,\,\,\,\,\,\,\,\,\,\,\,\,\,\,\,\,\,\,\,\,\,\,\,\,\,\,\,\,\,\,\,\,\,\,\,\,\,\,
 \times \prod_{i=1}^n\frac{\left(p_1^{(I^v_1(i))}(\mathbf{x}_{\sigma_1(i)})\right)^{\omega_1}\left(p_2^{(\theta(I^v_1(i)))}(\mathbf{x}_{\sigma_1(i)})\right)^{\omega_2}}{\int{\left(p_1^{(\bII^v_1(i))}(\mathbf{x}_{\sigma_1(i)})\right)^{\omega_1}\left(p_2^{(\theta(\bII^v_1(i)))}(\mathbf{x}_{\sigma_1(i)})\right)^{\omega_2} d\bx_{\sigma(i)}}}\\&
=\sum_{\sigma}\sum_{I_1\in \mathcal{F}_n(\mathbb{L}_1)}\sum_{\theta\in \Theta_{I_1}}\widetilde{w}_\omega^{(I_1,\theta)}\prod_{i=1}^{n}p_{\omega}^{(\bII^v(i),\theta)}(\mathbf{x}_{\sigma(i)})
\end{split}
\end{equation}
\hrulefill
 \end{figure*}
By substituting (\ref{approx-1-v2}) into (\ref{G-CI}), and utilizing Definition 1, the numerator of (\ref{G-CI}) can be rewritten as  (\ref{eq:1}), where $\widetilde{w}_\omega^{(I_1,\theta)}$ and $p_{\omega}^{(\ell,\theta)}(\bx)$ are shown in (\ref{fuse-w_p}) and (\ref{fuse-p_p}), respectively.

Thus, the denominator C of (\ref{G-CI}) can be computed as:
\vspace{-0.05in}
\begin{align}\label{K-1}
\begin{split}
C=&\int \widetilde{\pi}_\omega(\left\{\mathbf{x}_1,\ldots,\mathbf{x}_n\right\}) \delta\bX\\
=&\sum_{n=0}^{\infty}\sum_{I_1\in \mathcal{F}_n(\mathbb{L}_1)}\sum_{\theta\in\Theta_{I_1}}\widetilde{w}_\omega^{(I_1,\theta)}\\
=&\sum_{I_1\in \mathcal{F}(\mathbb{L}_1)}\sum_{\theta\in\Theta_{I_1}}\widetilde{w}_\omega^{(I_1,\theta)}.
\end{split}
\end{align}
Finally, by substituting (\ref{eq:1}) and (\ref{K-1}) into (\ref{G-CI}), we obtain the fused density as the form of (\ref{fuse-1}), which is a GMB  distribution, the unlabeled version of GLMB distribution \cite{LMB_Vo,LMB_Vo2,GCI-GMB}.
\end{proof}
\subsection{MB Approximation}
In Section III-A, we have approximated the fused distribution as a GMB distribution. In practical scenarios, the G-CI fusion in a sensor network is usually realized by sequentially applying the G-CI fusion rule \cite{Battistelli},
since a sensor network always has more than two sensors. In addition, in order to enhance the performance of a sensor network further, the feedback work mode is sometimes enabled. Thus, the fused posterior needs to be in the same form of the local posteriors, and it is necessary to  approximate the GMB formed fused posterior as an MB distribution. Motivated by \cite{MeMBer_Vo2} and \cite{delta_GLMB}, in which the multi-object distribution is approximated by exact moment matching, we  further seek an MB approximation that matches the first-order moment of the GMB  formed fused posterior in (\ref{fuse-1}).
\begin{Pro}\label{P3}
Suppose the fused posterior has been approximated as a GMB of form  (\ref{fuse-1}). 
The MB distribution that matches exactly the first-order moment of the fused posterior $\pi_{\omega}(\bX)$ is $\pi_{\textit{MB}}(\bX)=\{(r_\omega^{(\ell)},p_\omega^{(\ell)})\}_{\ell\in\mathbb{L}_1}$, where
\begin{align}
\label{r_l}&r_\omega^{(\ell)}=\sum_{I_1\in \mathcal{F}(\mathbb{L}_1)}\sum_{\theta\in\Theta_{I_1}}1_{I_1}(\ell)w_\omega^{(I_1,\theta)}\\
\label{p_l}&p_\omega^{(\ell)}(\mathbf{x})=\sum_{I_1\in \mathcal{F}(\mathbb{L}_1)}\sum_{\theta\in\Theta_{I_1}}1_{I_1}(\ell)w_\omega^{(I_1,\theta)}p_{\omega}^{(\ell,\theta)}(\mathbf{x})\bigg{/}r_\omega^{(\ell)}.
\end{align}
\end{Pro}
\begin{proof}
According to Proposition \ref{P2}, a GMB distribution shown in  (\ref{fuse-1}) is used to approximate the fused posterior of G-CI fusion with two MB distributions, and its first order moment can be computed as

\begin{equation}\label{fused-PHD}
\begin{split}
    &\,\,\,\,\,\,v(\mathbf{x}_1)\\
    &=\sum_{n=1}^{\infty}\frac{1}{(n\!-\!1)!}\int\widehat{\pi}_{\omega}(\{\mathbf{x}_1,\mathbf{x}_2,\cdots,\mathbf{x}_n\})d\mathbf{x}_2,\cdots,d\mathbf{x}_n\\
    &=\sum_{n=1}^{\infty}\frac{1}{( n\!-\!1)!}\!\sum_{I_1\in\mathcal{F}_n(\mathbb{L}_1)}\!\sum_{\theta\in\Theta_{I_1}}\!\sum_{\sigma}w_\omega^{(I_1,\theta)}p_\omega^{(\bII^v(\sigma^{\!-1}(1)),\theta)}(\bx_1)\\
    &=\sum_{n=1}^{\infty}\sum_{(I_1,\theta)\in\mathcal{F}_n(\mathbb{L}_1)\times\Theta_{I_1}}\sum_{\ell\in I_1}w_\omega^{(I_1,\theta)}p_\omega^{(\ell,\theta)}(\bx_1)\\
    &=\sum_{\ell\in\mathbb{L}_1}\sum_{I_1\in \mathcal{F}(\mathbb{L}_1)}\sum_{\theta\in\Theta_{I_1}}1_{I_1}(\ell)w_\omega^{(I_1,\theta)} p_{\omega}^{(\ell,\theta)}(\mathbf{x}_1)\\
 &=\sum_{\ell\in\mathbb{L}_1} r_\omega^{(\ell)} p_\omega^{(\ell)}(\mathbf{x}_1).
\end{split}
\end{equation}
Equation~(\ref{fused-PHD}) proves that the MB distribution with parameters $\{(r_\omega^{(\ell)}, p_\omega^{(\ell)})\}_{\ell\in\mathbb{L}_1}$ shown in (\ref{r_l}) and (\ref{p_l}) matches exactly the first-order statistical moment of the GMB distribution produced by~(\ref{fuse-1}).
\end{proof}
\begin{Rem}
\normalfont{Fusion (\ref{G-CI}) with MB densities can be easily extended to $N_s\geqslant2$ sensors by sequentially applying the pairwise fusion
(\ref{r_l}) and (\ref{p_l}) $N_s-1$ times, where the ordering of pairwise fusions is irrelevant. Similar approach has been widely used in distributed fusion, such as GCI fusion with CPHD filters \cite{Battistelli} and LMB filters \cite{Fantacci-BT}.
}
\end{Rem}
\begin{Rem}
\normalfont{To implement the GCI fusion with MB densities algorithm, we need to firstly compute the $Z_\omega^{(\ell,\theta)}$ and $p_\omega^{(\ell,
\theta)}$ under each hypothesis according to (\ref{Z_w}) and (\ref{fuse-p_p}), then compute the $r_\omega^{(\ell)}$ and $p_\omega^{(\ell)}(\bx)$ according to (\ref{r_l}) and (\ref{p_l}).
However, it can be seen from (\ref{r_l}) and (\ref{p_l}) that the number of hypotheses grows exponentially with the number of targets. In order to reduce the computational burden, we can perform truncation of the GMB density  using the ranked assignment strategy \cite{LMB_Vo,LMB_Vo2} or parallel filtering by grouping targets \cite{delta_GLMB}.}
\end{Rem}

\begin{figure}[ht]
\centering
\includegraphics[width=9.cm]{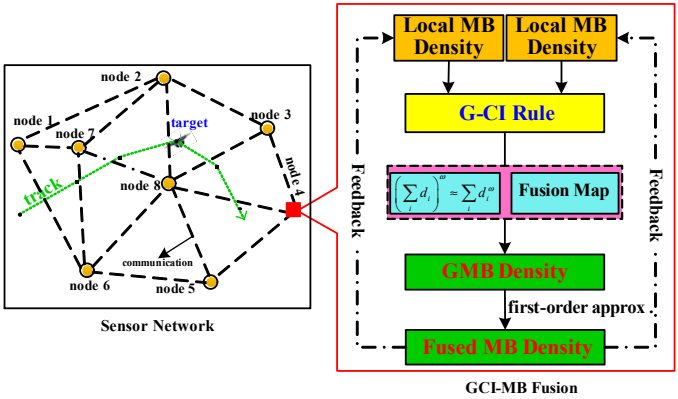} 
\caption{\small{ A sensor network is shown in the left. Each node monitors targets and exchange posterior with its neighbours. The proposed GCI-MB fusion algorihtm is employed to complete the distributed fusion task over the sensor network, and the process diagram of GCI-MB fusion algorithm is shown in the right.
}}
\label{fig_sim}
\end{figure}

\subsection{Summary}

In this section, we proposed a distributed multi-sensor multi-object tracking algorithm based on G-CI fusion rule and MB distribution, henceforward referred to as GCI-MB. By employing two reasonable approximations, the fused posterior density of two MB densities after GCI-MB is also an MB distribution. Therefore, by sequentially applying the closed form solution of GCI-MB, we can complete  another fusion process between the previously fused results and a third sensor node. In addition, the MB formed density after GCI-MB also facilitates the feedback process to further improve the fusion performance.

The process diagram of GCI-MB fusion is shown in right part of Fig. 3 and the proposed method is employed to complete the distributed fusion for a sensor network as shown in left part of Fig. 3. The complete GCI-MB fusion scheme for a sensor network includes the following steps:
 \begin{enumerate}
 \item Local filtering: each local sensor node runs MB filtering;
 \item Information exchange: each node exchanges its posteriors with their neighbors;
 \item Posterior fusing: each node performs GCI-MB fusion by sequentially applying the fused posterior.
 \item Feedback: to further improve the performance of the GCI-MB fusion, the fused MB density is fed back to each local node.
 \end{enumerate}

 Take node 4 as an example. In the first step, the local measurements are used to update a local MB posterior. The node then exchanges its posterior with nodes 3, 4, 5 and 8 and collects their posteriors. In the next step, it sequentially performs GCI-MB fusion three times, and finally at feedback stage, the fused posterior is fed back to local sensor nodes 3, 4, 5 and 8 to further improve the fusion performance.

\section{Implementation of GCI-MB Fusion Algorithm}
The conventional SMC implementation of MB  filter is used to compute the fusion of local information. To fuse information from different sensor nodes, we must be able to compute (\ref{r_l}) and (\ref{p_l}). However, this cannot be carried out directly because
each node has its own particle filter with its own support. Therefore, we use a a kernel density estimation (KDE)~\cite{Uney-2} method to create continuous approximations of the local posteriors. These posteriors are then sampled from to compute the G-CI fusion using different particle supports.

For the detail of SMC implementation of MB filter, the reader is referred to \cite{MeMBer_Vo2,MeMber_Vo3}. We present the SMC implementation of GCI-MB fusion directly.

\subsection{SMC Implementation of GCI-MB}
Let us denote the particle representation of each node's local MB distribution by $\left\{r_s^{(\ell)},\{\zeta_{s, m_s}^{(\ell)}, \bx_{s,m_s}^{(\ell)}\}_{m_s=1:L_s^{(\ell)}}\right\}_{\ell \in \mathbb{L }_s}$ with
\begin{align}\label{s_x}
\begin{split}
p_s^{(\ell)}(\bx)=\sum_{m_s=1}^{L_s^{(\ell)}}\zeta_{s,m_s}^{(\ell)}\delta_{\bx_{s,m_s}^{(\ell)}}(\bx), \,s=1,\,2
\end{split}
\end{align}
where
$\zeta_{s,m_s}^{(\ell)}$ is the weight associated with the $m_s$-th particle $\bx_{s,m_s}^{(\ell)}$ which is a point generated from the $\ell$-th density, and the $L_s^{(\ell)}$ denotes the number of particles representing the $\ell$-th density.

In Section III, we derived the closed-form expression of the fused posterior as an MB distribution with its MB parameters shown in Proposition~\ref{P3}. The implementation of
GCI-MB is equivalent to calculate the MB parameters of the fused posterior, including the existing probability  $r^{(\ell)}$ and its  density $p^{(\ell)}(\bx)$ conditional on existence, $\ell \in \mathbb{L}_1$. During the computing process, the parameters $p_{\omega}^{(\ell,\theta)}(\bx)$ and $Z_{\omega}^{(\ell,\theta)}$ in (\ref{fuse-1}), (\ref{r_l}) and (\ref{p_l})  are the key factors.

As it was mentioned earlier, two  posteriors presented by particles from two nodes cannot directly be fused via GCI-MB, for each node has its own set of particles. Neither the support nor the number of particles are guaranteed to be the same. Thus, we employ KDE, in which the estimated density is a sum of kernel function shifted to particle points. We associate each $p_s^{(\ell)}(\bx), \ell\in\mathbb{L}_s$ with the parameter $\mathbf{\Sigma}_s^{(\ell)}$ and use the density given by
\begin{align}\label{p_x_estimation}
\begin{split}
\widehat{p}_s^{(\ell)}(\bx)=\frac{1}{L_s^{(\ell)}}\sum_{m_s=1}^{L_s^{(\ell)}}\mathcal{N}\left(\bx;\bx_{s,m_s}^{(\ell)},\mathbf{\Sigma}_{s}^{(\ell)}\right),\ell\in \mathbb{L}_s
\end{split}
\end{align}
where $\mathcal{N}\left(\bx;\bx_{s, m_s}^{(\ell)},\mathbf{\Sigma}_{s}^{(\ell)}\right)$ is a Gaussian distribution with mean $\bx_{s,m_s}^{(\ell)}$ and covariance $\mathbf{\Sigma}_{s}^{(\ell)}$.

Next, we describe the computation of $\mathbf{\Sigma}_{s}^{(\ell)}$ for $p_s^{(\ell)}(x),\ell\in \mathbb{L}_s$:

In order to find the kernel parameters $\mathbf{\Sigma}_{s}^{(\ell)}$ for the members of the Bernoulli component $\ell$, we first find a transform that diagonalizes the empirical covariance of these points in the transformed domain. Then, the problem of finding the kernel parameters in multiple-dimensions reduces to independent single dimensional problems.

The transform is given by the inverse square root of the empirical covariance matrix $\Upsilon_\ell$ of Bernoulli component $\ell$. We transform all $\bx_{s,m^{\prime}}^{(\ell)}\in\{\bx_{s,m^{\prime}}^{(\ell)}|m^{\prime}=1,\ldots,L_s^{(\ell)}\}$ using
\begin{align}\label{transform-1}
\begin{split}
\mathbf{y}_{s,m^{\prime}}^{(\ell)}=\mathbf{W}_{\ell}\bx_{s,m^{\prime}}^{(\ell)}
\end{split}
\end{align}
\begin{align}\label{transform-2}
\begin{split}
\mathbf{W}_{\ell}=\Upsilon_{\ell}^{-1/2}
\end{split}
\end{align}
Given that the covariance of $\mathbf{y}_{s,m^{\prime}}^{(\ell)}$ is diagonal, the d$_{\emph{state}}$-dimensional Gaussian kernel in the transformed domain simplifies to
\begin{align}\label{transform-simple}
\begin{split}
K\left(\mathbf{y},\mathbf{y}_{s,m^{\prime}}^{(\ell)}\right)=\prod_{{d'}=1}^{d_{\emph{state}}}\frac{1}{\sqrt{2\pi}h_{d'}}\exp\left(-\frac{1}{2}\frac{(\mathbf{y}^{d'}-\mathbf{y}_{s,m^{\prime}}^{(\ell),{d'}})^{2}}{h_{d'}^{2}}\right).
\end{split}
\end{align}
where d$_{\emph{state}}$ is the dimensionality of the state space and $h_{d'}$s are the bandwidth (BW) parameters of the 1-D Gaussian kernels.

The BW $h_{d'}$ for each dimension can be found using one of the well established methods in the literature [35]. In particular, we use the following rule-of-thumb (RUT) \cite{RCT}:
\begin{align}\label{BW}
\begin{split}
h_{d'}=\sigma_{d'}\left(\frac{4}{3N}\right)^{1/5}
\end{split}
\end{align}
where $\sigma_{d'}$ is the empirical standard deviation of $y^d_{j}$s and $N$ is the number of these points. The reason for this choice is its simplicity and low computational complexity compared to other methods such as \textit{least squares cross-validation} \cite{Jones}.

The covariance matrix that specifies the kernels in (39) for the members of the Bernoulli component $\ell$ is given by
\begin{align}\label{parameters}
\begin{split}
&\mathbf{\Sigma}_s^{(\ell)}=\mathbf{T}_{\ell}\mathbf{\Lambda}_{\ell}\mathbf{T}_{\ell}^{T}\\
&\mathbf{T}_{\ell}=\mathbf{W}_{\ell}^{-1}\\
&\mathbf{\Lambda}_{\ell}=\mbox{diag}(h_1^2,h_2^2,\ldots,h_{d_{\emph{state}}}^2).
\end{split}
\end{align}
\subsubsection{Estimation of  the Parameter $p_w^{(\ell,\theta)}(\bx)$ }
The union of the input particle sets, i.e.,
\begin{align}\label{uinon of particles_implementation}
\begin{split}
P_U&\triangleq\{\bx_{1,m_1}^{(\ell)}\}_{m_1=1:L_1^{(\ell)}}\bigcup \{\bx_{2,m_2}^{(\theta(\ell))}\}_{m_2=1:L_2^{(\theta(\ell))}}
\end{split}
\end{align}
can be seen as $L_U=L_1^{(\ell)}+L_2^{(\theta(\ell))}$ samples drawn from the mixture important samping (IS) density
\begin{align}\label{p_IS}
\begin{split}
&p_{\emph{IS}}(\bx)=\frac{L_1^{(\ell)} p_1^{(\ell)}(\bx)^{\omega_1}+L_2^{(\theta(\ell))} p_2^{(\theta(\ell))}(\bx)^{\omega_2}}{L_1^{(\ell)} +L_2^{(\theta(\ell))}}
\end{split}
\end{align}

Therefore, $P_U$ given by (\ref{uinon of particles_implementation}) is a convenient particle set to represent $p_{\omega}^{(\ell,\theta)}(\bx)$ and the IS weights for $\bx_{m^{\prime}}\in P_U$ are given by
\begin{align}\label{weight_implementation}
\begin{split}
&\zeta_{m^{\prime}}\propto\frac{p_1^{(\ell)}(\bx_{m^\prime})^{\omega_1}p_2^{(\theta(\ell))}(\bx_{m^\prime})^{\omega_2}}
{L_1^{(\ell)}p_1^{(\ell)}(\bx_{m^\prime})^{\omega_1}+L_2^{(\theta(\ell))}p_2^{(\theta(\ell))}(\bx_{m^\prime})^{\omega_2}}.
\end{split}
\end{align}

In order to compute the IS weights in (\ref{weight_implementation}), evaluations of both $p_1^{(\ell)}(\bx_{m^\prime})$ and $ p_2^{(\theta(\ell))}(\bx_{m^\prime})$ at all points of
$P_U$ are necessary.
After obtaining the KDEs of $p_1^{(\ell)}(\bx_{m^\prime})$ and $ p_2^{(\theta(\ell))}(\bx_{m^\prime})$ using (\ref{p_x_estimation}) respectively, feasible estimates of $\widehat{\zeta}_{m^{\prime}}$s
 are computed by substituting these evaluations into (\ref{weight_implementation}):
\begin{align}\label{weight_estimation_implementation}
\begin{split}
&\widehat{\zeta}_{m^{\prime}}\propto\frac{\widehat{p}_1^{(\ell)}(\bx_{m^\prime})^{\omega_1}\widehat{p}_2^{(\theta(\ell))}(\bx_{m^\prime})^{\omega_2}}
{L_1^{(\ell)}\widehat{p}_1^{(\ell)}(\bx_{m^\prime})^{\omega_1}+L_2^{(\theta(\ell))}\widehat{p}_2^{(\theta(\ell))}(\bx_{m^\prime})^{\omega_2}}.
\end{split}
\end{align}
After resampling $\{\widehat{\zeta}_{m^{\prime}},\bx_{m^{\prime}}\}_{m^{\prime}=1:L_{U}}$, we obtain equally weighted samples to represent $p_\omega^{(\ell,\theta)}(\bx)$.
\subsubsection{Estimation of $Z_\omega^{(\ell,\theta)}$ }
Using the proposal density $p_{\emph{IS}}(\bx)$ given in (\ref{p_IS}), the IS estimate of  $Z_\omega^{(\ell,\theta)}$ is given by
\begin{align}\label{Z_estimation}
\begin{split}
&Z_\omega^{(\ell,\theta)}\triangleq \\
&\sum_{\bx_{m^\prime}\in P_U} \frac{p_1^{(\ell)}(\bx_{m^\prime})^{\omega_1}p_2^{(\theta(\ell))}(\bx_{m^\prime})^{\omega_2}}{L_1^{(\ell)} p_1^{(\ell)}(\bx_{m^\prime})^{\omega_1}+L_2^{(\theta(\ell))} p_2^{(\theta(\ell))}(\bx_{m^\prime})^{\omega_2}}
\end{split}
\end{align}
where $P_U$ is the union of the input particle  sets (\ref{uinon of particles_implementation}).

We substitute the KDEs of $\widehat{p}_1^{(\ell)}(\bx_{m'})$ and $\widehat{p}_2^{(\theta(\ell))}(\bx_{m'})$ into (\ref{Z_estimation}) to achieve computational feasibility and obtain
\begin{align}\label{Z_estimation_implementation}
\begin{split}
&\widehat{Z}_\omega^{(\ell,\theta)}\triangleq \\
&\sum_{\bx_{m^\prime}\in P_U} \frac{\widehat{p}_1^{(\ell)}(\bx_{m^\prime})^{\omega_1}\widehat{p}_2^{(\theta(\ell_1))}(\bx_{m^\prime})^{\omega_2}}{L_1^{(\ell)} \widehat{p}_1^{(\ell)}(\bx_{m^\prime})^{\omega_1}+L_2^{(\theta(\ell))} \widehat{p}_2^{(\theta(\ell))}(\bx_{m^\prime})^{\omega_2}}.
\end{split}
\end{align}

\subsection{ Pseudo-code }
A brief summary of the SMC implementation of GCI-MB is presented in the following algorithm. The first inputs of the algorithm are the particle sets (\ref{s_x}) of  local MB posteriors from both sensors.
\begin{enumerate}[Step 1.]
\item Under each $I_1\in\mathcal{F}(\mathbb{L}_1)$, create the map space $\Theta(I_1)$.

\item  Under each $(I_1,\theta)\in\mathcal{F}(\mathbb{L}_1)\times\Theta(I_1)$,  for each $\ell \in I_1, \theta(\ell)\in \theta(I_1)$:
                                                               \begin{enumerate}[$\bullet$]
                                                               \item Compute the KDE parameters of $\widehat p_1^{(\ell)}(\bx)$ and $\widehat p_2^{(\theta(\ell))}(\bx)$ in (\ref{p_x_estimation}), respectively;
                                                                \item According to (\ref{uinon of particles_implementation}), construct the sample set $P_U$ drawn from the IS density (\ref{p_IS});
                                                                \item Evaluate KDEs of  the input local densities at each particles in this set according to (\ref{p_x_estimation}).
                                                                \item Evaluate the IS weights for this sample set  according to (\ref{weight_implementation});

                                                         \item  Resample this sample set and  obtain the normalized weight of $p_\omega^{(\ell, \theta)}(\bx)$.
                                                               \item Evaluate the quantity $Z_\omega^{(\ell,\theta)}$ according to  (\ref{Z_estimation_implementation});

                                                               \end{enumerate}
\item Calculate the weight $w_\omega^{(I_1,\theta)}$ for each $(I_1,\theta)\in\mathcal{F}(\mathbb{L}_1)\times\Theta_{I_1}$ according to (\ref{fuse-w}), (\ref{fuse-w_p}) and (\ref{K-1}).
\item Calculate each fused MB parameter  $r^{(\ell)}$ and its density $p^{(\ell)}(\bx)$, $\ell\in\mathbb{L}_\omega$ according to (\ref{p_l}) and (\ref{r_l}).
\end{enumerate}
The output of the algorithm is a set of particles representing the fused posteriors with parameters $r^{(\ell)}$ and $p^{(\ell)}(\bx)$.

\section{Performance Assessment}
In this section, the performance of the proposed GCI-MB fusion algorithm is examined in two tracking scenarios
in terms of the optimal sub-pattern assignment (OSPA) error \cite{MeMBer_Vo1}. GCI-MB  is implemented   using the SMC approach proposed in Section IV. Since this paper does not focus on the problem of weight selection, we choose the Metropolis weights \cite{L_Xiao} in GCI-MB fusion for convenience (we note that this may have an impact on the fusion performance).

The MB filter for image data, also referred to as MB track-before-detect (MB-TBD) filter \cite{MeMber_Vo3} is used to estimate local sensors' posteriors. Local filters adopt the ``standard'' target motion model \cite{MeMBer_Mahler} without target births. Each single target with a four-dimensional state vector containing the two-dimensional positions and velocities is initialized within region around the correct target positions, and follows a constant velocity  model. The probability of survival is $p_{e}^k=0.95$. An  image observation model is used similarly as \cite{MeMber_Vo3}. The surveillance region is divided into $D$ resolution cells denoted as $V_{1},V_{2},\cdots,V_{D}\subset\mathbb{R}^{\nu/2}$. At time $k$, we represent the observations at time $k$ as $\bz^k=(z^k_{1},z^k_{2},\cdots,z^k_{D})^{\prime}\in \mathbb{R}^{D}$, with $z^k_{j}$ the  observation data obtained from the $j$th cell. A target with state $\mathbf{x}$ illuminates a set of pixels denoted by $U(\mathbf{x})$. Targets are assumed to be rigid bodies, which means that the regions affected by different targets do not overlap, i.e., $\mathbf{x}\neq \mathbf{x}^{\ast}\Rightarrow U(\mathbf{x})\cap U(\mathbf{x}^{\ast})=\emptyset$. Assuming that the values of different pixels are independently distributed conditioned on the multi-object state $\mathbf{X}^k$,  the multi-object likelihood function $g(\bZ^k|\bX^k)$ of $\bZ^k=\{\bz^k\}$ is given by:
\begin{equation}\label{multi-target likelihood}
  g(\mathbf{Z}^k|\mathbf{X}^k)=f(\bz^k)\prod_{\mathbf{x}\in \mathbf{X}^k }g_{z}(\mathbf{x})
  \end{equation}
where
\begin{align}
\notag g_{z}(\mathbf{x})&=\prod_{j\in U(\mathbf{x})}\frac{P_{H_1}(z^k_{j}|\mathbf{x})}{P_{H_0}(z^k_{j})}\\
\notag f(\bz^k)&=\prod_{j=1}^D P_{H_0}(z_{j}^k)
\end{align}
with $P_{H_1}(z^k_{j}|\mathbf{x})$ the observation density function for the $j$th cell occupied by the target state
 $\mathbf{x}$ and
 $P_{H_0}(z^k_{j})$ the noise density for the $j$th cell. For different applications,  $P_{H_1}(z_j^k|\bx)$ and $P_{H_0}(z_j^k|\bx)$ have different distributions, such as Gaussian distribution \cite{MeMber_Vo3}, Rayleigh distribution and   Compound-Gaussian distribution \cite{no-Gaussian-clutter}, etc.

In the following experiments, we consider a two-dimensional scenario over $50\times50$ resolution cells with cell lengths $\delta_x=\delta_y=1$ m.
The interval between the sensor observations is $T=1$ s.
The probability densities  of the intensity $z^k_{j}$ of pixel $j$, at  time $k$, adopt Gaussian distribution, namely,
\[\displaystyle{
\begin{split}
P_{H_1}(z^k_{j}|\mathbf{x})&=\mathcal{N}\left(z^k_j;\sum_{\bx\in X^k}\sigma_j^T(\bx),\sigma^N\right)\\
P_{H_0}(z^k_{j})&=\mathcal{N}\left(z^k_j; 0,\sigma^N\right)
\end{split}
}\]
where $\sigma_j^T(\bx)$ is the power contribution from target state $\bx$ to the $j$th cell and $\sigma^N$ is the noise power. Here, $\sigma^T_j(\bx)$ is described by a point spread function \cite{MeMber_Vo3}, for example,
\[\displaystyle{
  \sigma^T_j(\bx)=\frac{\delta_x\delta_y\sigma^T}{2\pi\sigma^2_b}\exp\left(-\frac{(\delta_x a-p_x)^2+(\delta_y b-p_y)^2}{2\sigma^2_b}\right)
}\]
  where $\sigma^T$ is the source intensity, $\sigma^2_b$ is the blurring factor, $(p_x,p_y)$ is the position of the state $\bx$,  and $j=(a,b)$ denotes the position of the $j$th cell in two-dimensionality image of the surveillance region. The SNR is defined by $10\log(\sigma^T/\sigma^N)$.  Here, the source intensity $\sigma^T$  is assumed to be the same deterministic  value for all the sensors. In practical scenarios, the $\sigma^T$ is always the random value and  follows different distributions among different sensors \cite{non-iid}, however, it is not the scope of this paper.
\subsection{Scenario 1}
Scenario 1 involves two parallel targets with the same  velocity as shown in Fig. \ref{fig_two_targets_scenario}, thus the $d_E$ between target states consisting of target position and velocity is completely determined by the physical  distance. For this scenario we apply a point spread function with the blurring factor $\delta^2_b=1$. $U(\mathbf{x})$ is the $3\times 3$ pixels square region whose center is closest to $(p_x,p_y)$. The SMC trials use $200$ particles per hypothesized track.
\begin{figure}[htbp]
\centering
\includegraphics[width=8cm]{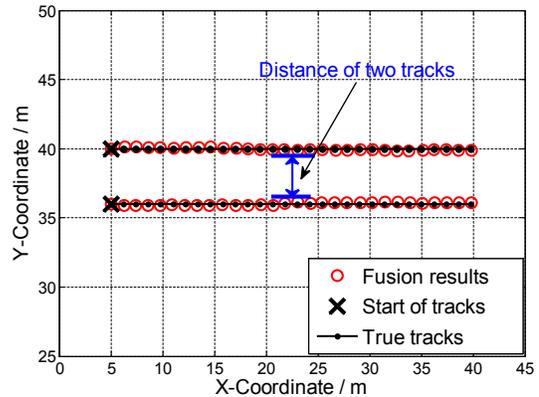} 
\caption{\small{ The scenario of distributed  sensor network with two sensors tracking two parallel targets.}}
\label{fig_two_targets_scenario}
\end{figure}

\subsubsection{Experiment 1}
Proposition \ref{P2}  applies the approximation in (\ref{fuse-2}) to obtain the GMB formed fused posterior. In Section III,  a Gaussian example had been provided to prove the reasonable of the approximation in (\ref{fuse-2}). In order to  back up that the approximation in (\ref{fuse-2}) is generalized enough to support the non-Gaussian case, the MB filter for image data is used to provide the multi-object estimations. Hence, we first examine  effectiveness of (\ref{fuse-2})  in terms of  the  absolute error between $\pi(\widehat{\bX})^{\omega}$ and $P(\widehat{\bX})$ at a given multi-target state estimation $\widehat{\bX}$, which is defined by
\begin{equation}\label{error}
\begin{split}
{E_{\{\pi^\omega,P\}}(\widehat{\bX})}=\big{|}{\pi(\widehat{\bX})}^{\omega}-P(\widehat{\bX})\big{|}
\end{split}
\end{equation}
where $\pi(\widehat{\bX})^{\omega}$ is in the form of (\ref{term_w}) and
\begin{equation}
\begin{split}
P(\widehat{\bX})= \sum_{\bII_h\in H(n)}\left({Q^{\bII_h}}\right)^{\omega}\left(\prod_{i=1}^{n}p^{(\bII_h(i))}(\widehat{\bx}_i)\right)^{\omega}
\end{split}
\end{equation}
The KDE method in  (\ref{p_x_estimation}) is adopted to  estimate the values of $\pi(\widehat{\bX})^{\omega}$ and $P(\widehat{\bX})$.

As mentioned in section III-A, the approximation in (\ref{fuse-2}) is mainly influenced by the target SNR and Euclidean distance between target states denoted by $d_E$ (this paper mainly focusses on the influence of $d_E$ and SNR).
Hence,  the approximation error of (\ref{fuse-2}) is evaluated by $E_{\{\pi^\omega,P\}}$ for different $d_E$ and different SNR values in this experiment. The physical distance $d_E$ varies within $2$--$6$ m and SNR varies from $6$ dB to $18$ dB.
\begin{figure}[htbp]
\centering
\includegraphics[width=8cm]{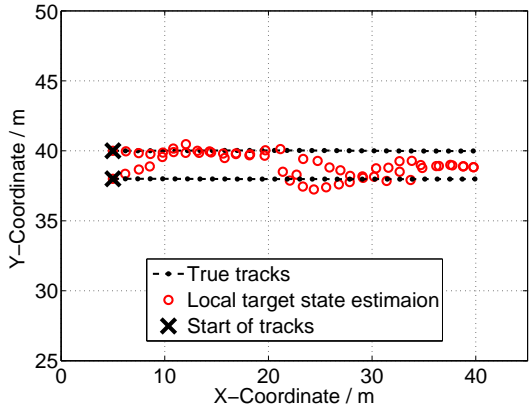} 
\caption{\small{ Target state estimation for two parallel targets with the $d_E$ between target states equals to 2m.}}
\label{fig_target_state_estimation}
\end{figure}
It is important to note that when targets are closely spaced, e.g. when $d_E=2~\text{m}, 3~\text{m}$, sometimes their state estimates may interfere with each other, making the distance between their estimations approach to $0$, as shown in Fig.~\ref{fig_target_state_estimation}. The detail analysis of this phenomenon is given in \cite{MeMber_Vo3}. Moreover, when this phenomenon arises,   the relationship between the approximation error
and the distance between target states cannot be reflected correctly. Hence, we compute a measure of  efficiency of estimation to evaluate  the validity of the average approximation error. When the OSPA of a multi-object state estimation is lower than a fixed value, we refer to this multi-object state as an efficient estimation. The proportion of efficient estimations is defined by
\begin{equation}
 \rho=\frac{N_{\emph{efficient}}}{N_{\emph{total}}}
 \end{equation}
where $N_{\emph{efficient}}$ is the number of efficient estimations, and $N_{total} $ is the total number of multi-object state estimations, which equals to the frame number times the number of Monte Carlo (MC) runs.
The approximate error is averaged over $N_{\emph{efficient}}$ estimations among 30 frames times 100 MC runs.

 \begin{figure}[htbp]
\begin{minipage}[htbp]{0.94\linewidth}
  \centering
  \centerline{\includegraphics[width=8cm]{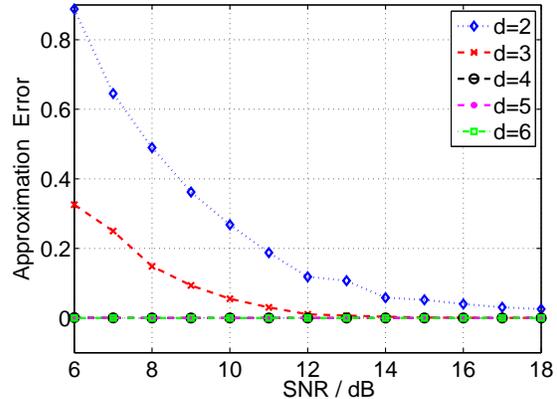}}
  \centerline{\small{\small{(a)}}}\medskip
\end{minipage}
\vfill
\begin{minipage}[htbp]{0.94\linewidth}
  \centering
  \centerline{\includegraphics[width=8cm]{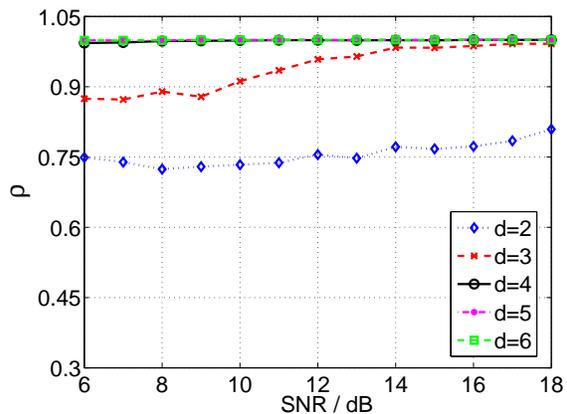}}
  \centerline{\small{\small{(b)}}}\medskip
\end{minipage}
\caption{(a)  The approximation error between (\ref{term_w}) and (\ref{fuse-2}) varies $d_E$ between true target states and  SNRs for efficient estimations, (b) the proportion of efficient estimations among 30 frames times 100 MC runs varies $d_E$ between true target states and  SNRs.}
\label{fig_approximation_error}
\end{figure}

 The approximation error under SMC implementation of MB filter is shown in Fig. \ref{fig_approximation_error} (a). The approximation error  is observed to be smaller with the bigger distances $d_E$ between the target states. More specifically, the approximation is generally acceptable when $d_E\geqslant3$ m, especially when $d_E\geqslant4$ m, the approximation error is very close to $0$ for all investigated SNRs ranging from $6$ dB to $18$ dB.  Thus the HPD regions are proved to be separated when $d_E\geqslant3$ m. The results also suggest that for a larger SNR, the approximation error is small even when the targets are in proximity , e.g. $d_E=3$ m. The larger SNR will lead to the smaller width of the HPD regions, and thus the smaller $d_E$ between target states could be tolerated.

 To further supplement the reasonableness of approximation (\ref{fuse-2}) for non-Gaussian cases, we provide the proportion of efficient estimations in Fig. \ref{fig_approximation_error} (b). It can be seen that the proportion of efficient estimations is close to $1$ (which means that the target state estimates are reliable) for large SNRs. Indeed, this occurs when $d_E\geqslant4$ m with SNR $\geqslant 6$ dB or when $d_E=3$ m with SNR $\geqslant 14$ dB. Overall, the results shown in Fig.~\ref{fig_approximation_error} (a) show that the approximation (\ref{fuse-2}) is acceptable when the estimations of target states are efficient and the approximation (\ref{fuse-2}) can be applied to perform  fusion for practical scenarios with the above conditions. These results conform that the approximation in (\ref{fuse-2}) is generalized enough to support the non-Gaussian case.

 As expected, the above discussions are in accordance with the analyses in Section III-A. Thus we come to conclusion that when the $d_E$ between target  states meets the targets separated condition ($d_E\geqslant3$ m for this simulation scenario), or the SNRs are large although they are nearly in proximity, the approximation (\ref{fuse-2}) is acceptable.  In addition, the approximation is more sensitive to the distance $d_E$ between target states than to the SNR.

\subsubsection{Experiment 2}
To prove the effectiveness  of  the GMB approximation in  (\ref{fuse-1}) described in Section III-A, we examine the sensor fusion performance for two sensors under different
$d_E$ between target states consisting of target position and velocity. SNR is fixed at 15 dB in order to reduce its influence on approximation (\ref{fuse-2}). The distance $d_E$ varies within $\{1, 2, 3, 4, 8, 12\}$ m. The SMC trials use $200$ particles per hypothesized track. The OSPA errors are averaged over 30 frames and 100 MC runs.
\begin{figure}[htbp]
\centering
\includegraphics[width=8cm]{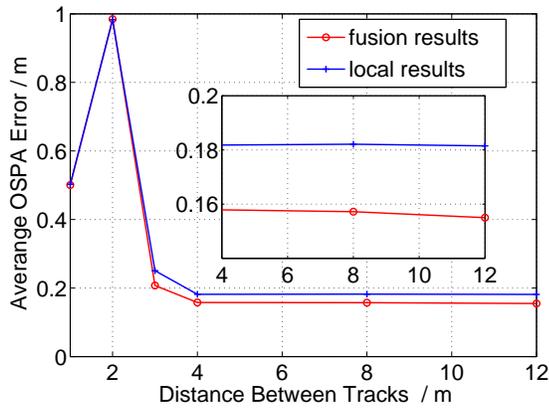} 
\caption{\small{ The average OSPA error with different distance of tracks (averaged over 100 MC runs).}}
\label{fig_OSPA_two_targets}
\end{figure}
Fig. \ref{fig_OSPA_two_targets} shows the average OSPA errors  for both the local filter and GCI-MB fusion algorithm versus the distance $d_E$.
It can be seen that the performance of GCI-MB fusion is better than local MB-TBD filter at each value of $d_E$. More specifically, the performance gains of fusion algorithm
are stable when $d_E\geq3$ m. When $d_E\leq2$ m, both algorithms perform poorly almost at the same level.  The reason is that when targets are in proximity thereby violating
the  rigid targets assumption, the performance of MB-TBD degrades heavily leading to the performance degradation of the fusion algorithm.  This is also the reason why we use the efficient estimations to evaluate the error of approximation (\ref{fuse-2}) in {Experiment 1}. Also note that  the fusion and tracking performances at $d_E=1$ m seem better than those at $d_E=2$ m  because  the state estimates are prone to be the middle of the two tracks when  the regions illuminated by different objects exhibit the superposition, and thus  the estimates at $d_E=1$ m are  nearer to the true target states than those at  $d_E=2$ m.  In summary, the above results verify that the GMB approximation is reasonable and effective when $d_E\geq3$ m.

\begin{Rem}
The  required minimum $d_E$ is equal to 3 times the cell resolution of the sensor network in the above experimental scenario, which is comparable to the regions affected by targets. Indeed, our experience with empirical data suggests that the distance between targets is mostly larger than 3 times the cell resolution in most practical scenarios.
\end{Rem}
\subsection{Scenario 2}
To assess the efficacy of the proposed GCI-MB fusion,
a sensor network scenario involving three targets  is considered as shown in Fig. \ref{fig_three_target_scenario}.  In this scenario, we apply a point spread function with the blurring factor $\delta^2_b=1$. $U(\mathbf{x})$ is the $5\times 5$ pixels square region whose center is closest to $(p_x,p_y)$.
\begin{figure}[htbp]
\centering
\includegraphics[width=9cm,height=6cm]{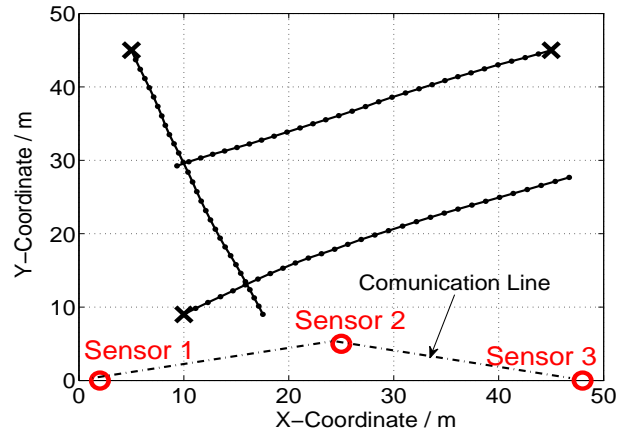} 
\caption{\small{ The scenario of distributed  sensor network with three sensors tracking three targets.}}
\label{fig_three_target_scenario}
\end{figure}
In this sensor network, each sensor has the same quality and can only exchange  posteriors with its neighbours. In particular, both sensors 1 and 3 perform GCI fusion with two posteriors from sensor 2 and the local filter, while sensor 2 performs GCI fusion with three posteriors from sensor 1, sensor 3 and the local filter by applying the pairwise fusion (\ref{r_l}) twice. There are two work modes in this sensor network given as follow:\\

 \begin{enumerate}[\textbf{M}1:]
 \item At time $k$ each sensor performs filtering locally, resulting in a local posterior denoted by $f_{\bM1}^l$. After receiving posteriors from its neighbours, it operates GCI fusion leading to fused posterior denoted by $f_{\bM1}^w$.
 \item  At time $k\!-\!1$, the fused posteriors are fed back to corresponding local filters. Then at time $k$, each sensor operates the local filter on the local distribution denoted by $f_{\bM2}^l$ and operate the GCI fusion on the fused one denoted by $f_{\bM2}^w$.
 \end{enumerate}

\subsubsection{Experiment 1}
In this experiment, the performance of the GCI-MB fusion is evaluated by comparison with that of the local MB-TBD filters in two work modes, and how the performance advantage gained from sensor fusion increases with more sensors is also provided. In the sensor network, each sensor choose the MB-TBD filter as the local filters.  The SMC trials use $200$ particles per hypothesized track.

Figs.~\ref{fig_OSPA_three_targets_1} and \ref{fig_OSPA_three_targets_2} show the OSPA errors of  both  the local filter and the GCI-MB fusion for sensor 1 and sensor 2 working in modes \textbf{ M}1 and \textbf{M}2. For sensor 3, results similar to sensor 1 are expected.

\begin{figure}[htbp]
\centering
\includegraphics[width=9cm]{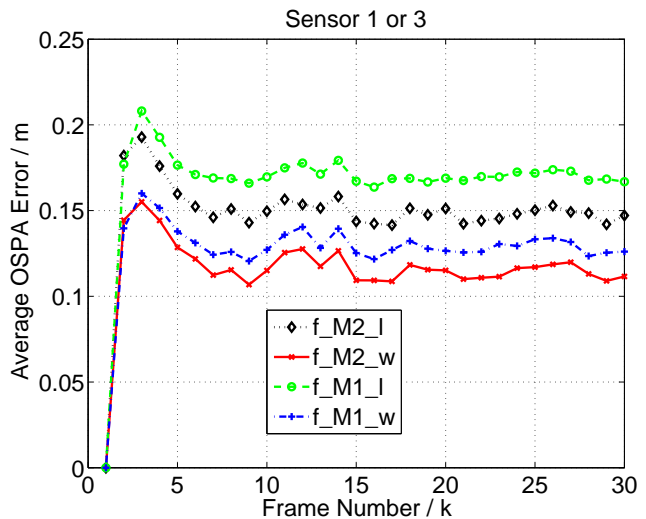} 
\caption{\small{ The average OSPA error of sensor 1 (averaged over 400 MC runs).}}
\label{fig_OSPA_three_targets_1}
\end{figure}
\begin{figure}[htbp]
\centering
\includegraphics[width=9cm]{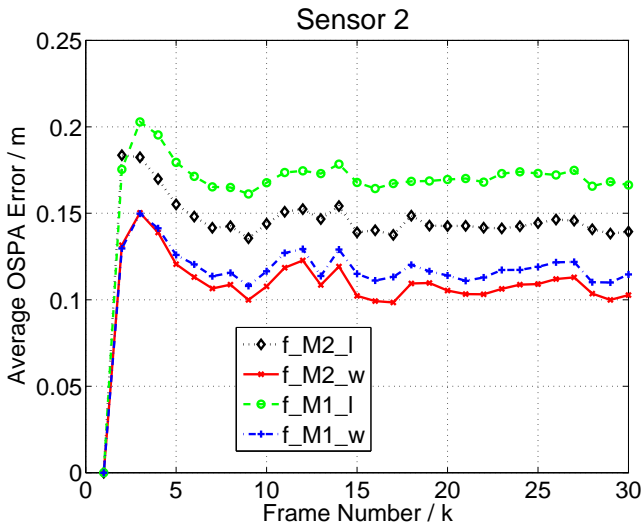} 
\caption{\small{ The average OSPA error of sensor 2 (averaged over 400 MC runs).}}
\label{fig_OSPA_three_targets_2}
\end{figure}
It can be seen from figures \ref{fig_OSPA_three_targets_1} and \ref{fig_OSPA_three_targets_2} that when the local filters receive feedback from  GCI-MB fusion ($\textbf{M}$2), they perform significantly better than $\textbf{M}$1. The theoretical analysis of the performance gain of the feedback on the fusion can reference \cite{DMMT-feedback1,DMMT-feedback2}. The significant enhancement in performance (in terms of OSPA errors) also verifies the effectiveness of  MB approximation and the GCI-fusion devised and presented in this work. To demonstrate how the performance advantage gained from sensor fusion increases with more sensors, we computed the OSPA errors averaged over 400 MC runs and 30 frames, and compared the results for the case when there is one sensor only, with the case of two sensors and the case of three sensors. In each case, both modes $\textbf{M}$1 and $\textbf{M}$2 were examined. The results are presented in Table~\ref{tabone} and demonstrate the efficacy of the proposed sensor fusion algorithm in the form of the enhanced average errors achieved with more sensors.

\begin{table}[htbp]
\renewcommand{\arraystretch}{1.5}
\caption{Average OSPA Error \textit{VS}  Number of Sensors}
\begin{center}
\begin{tabular*}{0.45\textwidth}{@{\extracolsep{\fill}}c c c c}
\toprule
 Number of sensor                         & One                  &   Two   &  Three  \\
\midrule
OSPA of M1 (m)                   &  0.1715               &    0.1282   &   0.1163  \\
OSPA of M2 (m)                   &  0.1715                & 0.1166     &  0.1093   \\
\bottomrule
\end{tabular*}
\label{tabone}
\end{center}
\end{table}
\subsubsection{Experiment 2}
In order to further demonstrate the utility of the proposed GCI-MB fusion, its performance is compared with the GCI fusion with PHD filter (GCI-PHD) proposed in \cite{Uney-2}. For local sensors, the PHD-TBD filter proposed in \cite{PHD_TBD} and the MB-TBD  filter are adopted in the GCI-PHD fusion and GCI-MB fusion respectively. The number of particles for PHD-TBD filter is 600, while the number of particles is 200   per hypothesized  track in the MB-TBD filter. Other parameters  are set to be  the same for PHD-TBD and MB-TBD filters.
\begin{figure}[htbp]
\centering
\includegraphics[width=9cm]{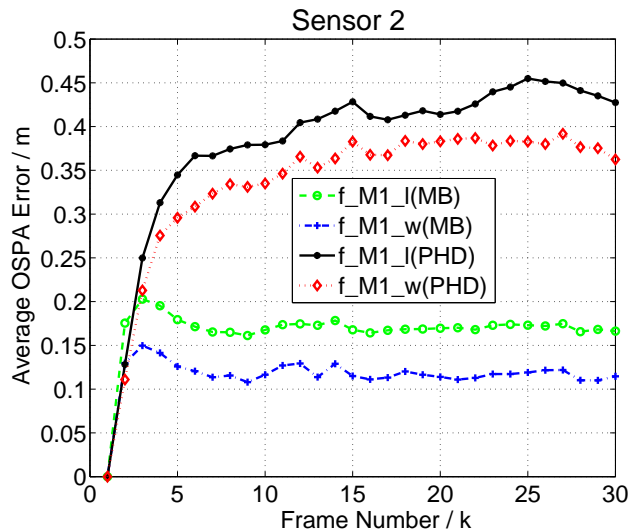} 
\caption{\small{The performance comparison between GCI-PHD and GCI-MB fusion at sensor 2 (averaged over 400 MC runs).}}
\label{fig_OSPA_GCI_MB_vs_PHD}
\end{figure}

Fig. 12  shows the OSPA errors of the local PHD-TBD filter, the GCI-PHD fusion, the local MB-TBD filter and  the GCI-PHD for  sensor 2 working in
\textbf{M}1. The curves shown in Fig.~\ref{fig_OSPA_GCI_MB_vs_PHD} illustrate the performance difference between  GCI-MB fusion  and GCI-PHD fusion, and their corresponding local filters, respectively. It can be seen that when the performances of  tracking or fusion algorithms reach a stable level, the OSPA error of the GCI-MB fusion is significantly lower  than the GCI-PHD fusion, and the similar performance difference  can be observed from their local filters. The reason is that the MB-TBD filter is a closed-form solution for the TBD observation model while the PHD-TBD filter is an approximate solution.  These results highlight the utility of the proposed GCI-MB fusion algorithm.
\section{Conclusion}
This paper investigates the problem of distributed multi-object tracking (DMMT) with multi-object multi-Bernoulli (MB) filter based on generalized Covariance Intersection. By employing two reasonable approximations, a tractable closed-form
formulation of GCI fusion with MB posteriors (GCI-MB) is derived. A particle implementation of the proposed GCI-MB fusion is also given, and its efficacy and robustness are demonstrated in numerical results. Future work will tack two major issues. Firstly, the number of hypotheses to be accounted for in the proposed distributed tracking solution grows exponentially with the number of targets. Further research is needed to investigate efficient implementations of GCI-MB in which irrelevant hypotheses are detected and pruned early, so the computational cost is limited in presence of numerous targets. Secondly, if targets move to close proximity of each other, Bernoulli components of the posteriors may not be well-separated. This will have an impact on the accuracy of the approximation made in derivation of the GCI-MB fusion. Further work will  address the GCI-MB fusion problem considering targets in proximity.


\begin{thebibliography}{1}

\bibitem{CY-Chong}
C.~Y. Chong, S. Mori, and K.~C. Chang, ``Distributed multitarget multisensor tracking,'' {\em Multitarget-Multisensor Tracking: Advanced Applications;
Y. Bar-Shalom (ed.); Artech House Chapter 8}, 1990.
\bibitem{EMD-Julier}
S. J. Julier, T. Bailey, and J. K. Uhlmann, ``Using exponential
mixture models for suboptimal distributed data fusion,'' in { \em Proc. IEEE Nonlinear Stat. Signal Proc. Workshop}, pp. 160-163, Sep. 2006.
\bibitem{Mahler-1}
R. Mahler, ``Optimal/Robust distributed data fusion: a unified approach,'' in {\em Proc. SPIE Defense Sec. Symp.}, 2000.

\bibitem{Uhlmann}
J. K. Uhlmann, ``Dynamic Map Building and Localization for Autonomous vehicles,'', { \em Ph.D. dissertation, Univ. of Oxford, Oxford, U.K.}, 1995.


\bibitem{chernoff-fusion}
K.~C. Chang, C.~Y. Chong, and S. Mori, ``Analytical and computational evaluation of scalscalable distributed fusion algorithms,''  {\em IEEE Trans. Aerosp.
Electron. Syst.}, Vol. 46, No. 4, pp. 2022-2034, Oct. 2010.
\bibitem{GMD-fusion}
T. Bailey, S. Julier, and G. Agamennoni, ``On conservative fusion of information with unknown non-Gaussian dependence,'' in {\em Proc. IEEE Int. Fusion Conf.}, pp. 1876-1883, Jul. 2012.
\bibitem{Clark}
D. Clark, S. Julier, R. Mahler, and B. Risti\'{c}, ``Robust multi-object sensor fusion with unknown correlations,'' in {\em Proc. Sens. Signal Process. Defence (SSPD ¡¯10)}, Sep. 2010.




\bibitem{Uney-2}
M. \"{U}ney, D. Clark, and S. Julier, ``Distributed fusion of PHD filters via exponential mixture densities,'' {\em IEEE J. Sel. Topics Signal Process.}, Vol. 7, No. 3, pp. 521-531, Apr. 2013.

\bibitem{Battistelli}
G. Battistelli, L. Chisci, C. Fantacci, A. Farina, and A. Graziano, ``Consensus CPHD filter for distributed multitarget tracking,'' {\em IEEE J. Sel. Topics Signal Process.}, Vol. 7, No. 3, pp. 508-520, Mar. 2013.

\bibitem{Mehmet}
M. B. Guldogan, ``Consensus Bernoulli filter for distributed detection and tracking using multi-static doppler shifts,'' {\em IEEE Signal Process. Lett.}, Vol. 21, No. 6, pp. 672-676, Jun. 2014.


\bibitem{MeMBer_Mahler}
 R.~P.~S. Mahler, Statistical Multisource-Multitarget Information Fusion. Norwell, MA, USA: Artech House, 2007.

\bibitem{P-Brcac}
P. Braca, S. Marano, V. Matta and P. Willett, ``Asymptotic effiency of the PHD in multitarget/multisensor estimation,'' { \em IEEE J. Sel. Top. Signal Process.}, Vol. 7, No. 3, pp. 553-564, Jun. 2013.

\bibitem{PHD-Vo}
B.-N. Vo and W.~K. Ma, ``The Gaussian mixture probability hypothesis density filter,'' {\em IEEE Trans. on Signal Process.}, Vol. 54, No. 11, pp. 4091-4104, Nov. 2006.

\bibitem{PHD_TBD}
K.~Punithakumar, T.~Kirubarajan, and A. Sinha, ``A sequential Monte
Carlo probability hypothesis density algorithm for multitarget track-before-detect,'' in {\em Proc. SPIE Conf. Signal Data Processing Small Targets,} San Diego, CA, vol. 5913, Aug. 2005.

\bibitem{Vo-CPHD}
B.~T. Vo, B.~N. Vo, and A. Cantoni, ``Analytic implementations of the cardinalized probability hypothesis density filter,'' { \em IEEE Trans. on Signal Process.}, Vol. 55, No. 7, pp. 3553-3567, Jul. 2007.


\bibitem{MeMBer_Vo2}
B.~T. Vo, B.~N. Vo, and A. Cantoni, ``The cardinality balanced multi-target multi-Bernoulli filter and its implementations,'' {\em IEEE Trans. on Signal Process.},
Vol. 57, No. 2, pp. 409-423, Oct. 2009.
\bibitem{MeMber_Vo3}
B.~T. Vo, B.~N. Vo, N.~T. Pham and D. Suter, ``Joint detection and estimation of multiple Objects from image observation,''  \emph{IEEE Trans. on Signal Process.},
Vol. 58, No. 10, pp. 5129-5141, Oct. 2010.
\bibitem{Vo_radar_target_nonline}
B. T. Vo, B.~N. Vo and R. Hoseinnezhad, ``Robust multi-Bernoulli filtering,''  {\em IEEE J. Sel. Topics Signal Process.} Vol. 7, No. 3, pp. 399-409, Jun. 2013.

\bibitem{Gunes_sonar}
A. Gunes and M.B. Guldogan, ``Multi-target bearing tracking with a single acoustic vector sensor based on multi-Bernoulli filter'' in {\em Proc. OCEANS}, Genova, Italy, pp. 1-5, May 2015.
\bibitem{Reza_ground_target}
B.~T. Vo, B.~N. Vo and R. Hoseinnezhad, ``Multi-Bernoulli based track-before-detect with road constraints,''  in {\em Proc. IEEE Int. Fusion Conf.}, pp. 840-846, Jul. 2012.

\bibitem{Reza_sensor_control_letter}
K.G. Amirali, R. Hoseinnezhad and B.H. Alireza, ``Robust multi-Bernoulli sensor selection for multi-target tacking in sensor networks,''  {\em IEEE Signal Process. Lett.}, Vol. 20, No. 12, pp.1167-1170, Dec. 2013.

\bibitem{Reza_sensor_control_AES}
K.G. Amirali, R. Hoseinnezhad and B.H. Alireza, ``Multi-Bernoulli sensor control via minimization of expected estimation errors,'' {\em IEEE Trans. Aerosp.
Electron. Syst.},  Vol. 51, No. 3, pp. 1762-1773, Jul. 2015.
\bibitem{Reza_audio_visual}
R. Hoseinnezhad, B.~N. Vo, B. T. Vo, and D. Suter, ``Bayesian integration of audio and visual information for multi-target tracking using a CB-MEMBER filter,'' in {\em Proc. Int. Conf. Acoust., Speech, Signal Process. (ICASSP), Prague, Czech Republic,} pp. 2300-2303, May 2011.
\bibitem{Reza_visual_tracking}
R. Hoseinnezhad, B.~N. Vo and B. T. Vo, ``Visual tracking in background subtracted image sequences via multi-Bernoulli filtering,''   \emph{IEEE Trans. on Signal Process.},  pp: 392-397,Vol. 61, No. 2, Jan. 2013.

\bibitem{MeMber_Wei1}
J. Wei and X. Zhang, ``Mobile multi-target tracking in two-tier hierarchical
wireless sensor networks,'' in {\em  Proc. IEEE Military Commun.
Conf.}, pp. 1-6, 2009.
\bibitem{MeMber_Wei2}
J. Wei and X. Zhang, ``Sensor self-organization for mobile multi-target
tracking in decentralized wireless sensor networks,'' in {\em Proc. IEEE
Wireless Commun. Netw. Conf.}, pp. 1-6, 2010.
\bibitem{Dunne_MM}
D. Dunne and T. Kirubarajan, ``Multiple model multi-Bernoulli filter for manoeuvring targets,'' {\em  IEEE Trans. Aerosp.
Electron. Syst.}, Vol. 49, No. 4, pp. 2679-2692, Oct. 2013.
\bibitem{Williams}
J. L. Williams, ``Hybrid Poisson and multi-Bernoulli filters'' in {\em Proc. IEEE Int. Fusion Conf.}, pp. 1103-1110, Jul. 2012.
\bibitem{LMB_Vo}
B.~N. Vo, B.~T. Vo, ``Labeled random finite sets and multi-object conjugate priors,'' {\em IEEE Trans. on Signal Process.}, Vol. 61, No. 10, pp. 3460-3475, Jul. 2013.

\bibitem{LMB_Vo2}
B.~N. Vo, B.~T. Vo, and D. Phung, ``Labeled random finite sets and
the Bayes multi-target tracking filter,'' \emph{IEEE Trans. on Signal Process.}, Vol.PP, No.99, pp.1, Oct. 2014.


\bibitem{delta_GLMB}
S. Reuter, B.~T. Vo, B.~N. Vo, and K. Dietmayer, ``The labeled multi-Bernoulli filter,'' \emph{IEEE Trans. on Signal Process.}, Vol. 62, No. 12, pp.3246-3260, Jun. 2014.

\bibitem{LRFS_Bear_Vo}
M.~Beard, B.-T. Vo, and B.-N. Vo, ``Bayesian multi-target tracking with merged measurements using labelled random finite sets,'' \emph{IEEE Trans. Signal Process.}, vol.~63, no.~6, pp. 1433--1447, 2015.

\bibitem{LRFS_Papi_Kim}
F.~Papi and D.~Y. Kim, ``A particle multi-target tracker for superpositional measurements using labeled random finite sets,'' \emph{IEEE Trans. Signal Process.}, vol.~63, no.~16, pp. 4348--4358, 2015.

\bibitem{Fantacci-BN}
C.~Fantacci, B.-T. Vo, F.~Papi, and B.-N. Vo, ``The marginalized $\delta$-glmb filter,'' \emph{arXiv preprint arXiv:1501.00926}, 2015.

\bibitem{GLMB_Papi_Vo}
F.~Papi, B.-N. Vo, B.-T. Vo, C.~Fantacci, and M.~Beard, ``Generalized labeled multi-bernoulli approximation of multi-object densities,'' \emph{IEEE Trans. Signal Process.}, vol.~63, no.~20, pp. 5487--5497, 2015.
\bibitem{Fantacci-BT}
C.~Fantacci, B.-N. Vo, B.-T. Vo, G.~Battistelli, and L.~Chisci, ``Consensus
  labeled random finite set filtering for distributed multi-object tracking,''
  \emph{arXiv preprint arXiv:1501.01579}, 2015.
  \bibitem{GCI-GMB}
B.~L. Wang, W. Yi, S. Q. Li, M. R. Morelande, L. J. Kong and X. B. Yang,  ``Distributed multi-target tracking via generalized multi-Bernoulli random finite sets,''  in {\em Proc. IEEE Int. Fusion Conf.}, pp. 253-261, Jul. 2015.
\bibitem{GCI-LSM}
B.~L.~Wang, W.~Yi, S.~Q.~Li, L.~J.~Kong and X. B. Yang, ``Distributed fusion of labeled multi-object densities via label spaces matching,'' \emph{arXiv preprint arXiv:1603.08336}, 2016.
\bibitem{Set_JPDA}
L. Svensson, D.~Svensson, and M. Guerriero, ``Set JPDA filter for multitarget tracking,'' {\em IEEE Trans. on Signal Process.}, Vol.59, No. 10, pp. 4677-4691, Aug. 2011.

\bibitem{KDE-1}
B. Silverman, ``Density estimation for statistics and data analysis,'' {\em London, U.K.: Chapman \& Hall}, 1986.

\bibitem{KDE-2}
C. Fraley and A. W. Raftery, ``Model based clustering, discriminant
analysis, and density estimation,'' {\em J. Amer. Statist. Assoc.}, Vol. 97, No.
458, pp. 611-631, Jun. 2002.
  \bibitem{GCI-MB-Con}
B.~L. Wang, W. Yi, S. Q. Li, L. J. Kong and X. B. Yang,  ``Distributed fusion with multi-Bernoulli based on generalized covariance intersection,''  in {\em Proc. IEEE Int. Radar Conf.}, pp. 958-962, May. 2015.

\bibitem{July}
S. J. Julier, ``An empirical study into the use of chernoff information
for robust, distributed fusion of Gaussian mixture models,'' in {\em Proc. IEEE Int. Fusion Conf.}, pp. 1-8, Jul. 2006.
\bibitem{HPD}
G.~E.~P.~Box and G.~C.~Tiao, Bayesian Inference in Statistical Analysis,  Addison Wesley, 1973.
\bibitem{RCT}
B. Silverman, {\em Density Estimation for Statistics and Data Analysis.}
London, U.K.: Chapman \& Hall, 1986.
\bibitem{Jones}
M. C. Jones, J. S. Marron, and S. J. Sheather, ``A brief survey of badwidth
selection for density estimation,''  {\em J. Amer. Statist. Assoc.}, Vol.
91, No. 433, pp. 401-407, Mar. 1996.




\bibitem{DMMT-feedback1}
W. Khawsuk and L. Y. Pao. ``Distributed multi-sensor multi-
target with feedback,'' in {\em Proc. American Control Conf}, Boston. MA, June 2004.

\bibitem{DMMT-feedback2}
 Y.~Zhu, J.~Zhao, K.~Zhang, X.~R.~Li, and Z.~You.``Performance
analysis for distributed track fusion with feedback,'' in {\em Proc. World
Congress on Intelligent Control and Automation}, pp. 2933-2936, 2000.


\bibitem{MeMBer_Vo1}
D. Schumacher, B.~T. Vo, B.~N. Vo, ``A consistent metric for performance evaluation of multi-object filters,'' {\em IEEE Trans. on Signal Process.}, Vol.56, No. 8, pp. 3447-3457, Aug. 2008.

\bibitem{L_Xiao}
L. Xiao, L. S. Boyd, and S. Lall, ``A scheme for robust distributed
sensor fusion based on average consensus,'' in  {\em Proc. 4th Int. Symp. Inf.
Process. Sens. Netw. (IPSN)}, 2005, pp. 63-70.
\bibitem{no-Gaussian-clutter}
G. L. Cui, L. J. Kong, X. B. Yang, ``Multiple-input multiple-output radar detectors design in non-Gaussian clutter,'' { \em IET Radar Sonar  Navig.,} Vol. 4, No. 5, pp. 724-732.


\bibitem{non-iid}
G. L. Cui, A. D. Maio, M. Piezzo, ``Performance prediction of the incoherent radar detector for correlated generalized Swerling-Chi fluctuating targets,''
{\em  IEEE Trans. Aerosp. Electron. Syst.,} Vol. 49, No. 1, pp. 356-368
 \end{thebibliography}
\end{document}